\documentclass{sig-alternate}
\pdfoutput=1

\usepackage{lmodern}
\usepackage[utf8]{inputenc}
\usepackage[TS1,T1]{fontenc}
\usepackage[usenames,dvipsnames,svgnames,table]{xcolor}
\usepackage{url}
\usepackage{graphicx}
\DeclareGraphicsExtensions{.eps}
\usepackage{latexsym}
\usepackage{xspace}
\usepackage{enumitem}
\usepackage{subcaption}
\usepackage{dblfloatfix}
\usepackage{flushend}
\usepackage{algorithm}
\usepackage{algorithmic}
\usepackage{amsmath,nccmath}
\usepackage{bm}
\usepackage{listings}
\usepackage{cleveref}

\setenumerate{noitemsep,topsep=0pt,parsep=0pt,partopsep=0pt}

\newtheorem{definition}{Definition}[section]

\newtheorem{theorem}{Theorem}

\newcommand*{\eat}[1]{}

\newcommand*{\ql}{\textsf{Portal}\xspace}
\newcommand*{\tg}{\textsf{TGraph}\xspace}
\newcommand*{\tgs}{\textsf{TGraphs}\xspace}

\newcommand*{\sg}{RG\xspace}
\newcommand*{\rg}{RG\xspace}

\newcommand*{\og}{OG\xspace}

\newcommand*{\hg}{HG\xspace}

\def\bp{\mathbf{p}}
\def\bc{\mathbf{c}}

\newcommand*{\edgec}{\textsf{edge-complete}\xspace}
\newcommand*{\vertexc}{\textsf{vertex-complete}\xspace}
\newcommand*{\edgeq}{\textsf{edge-query}\xspace}
\newcommand*{\vertexq}{\textsf{vertex-query}\xspace}

\newcommand*{\ttt}{\ensuremath{\mathsf{T}}\xspace}

\newcommand*{\tve}{\ensuremath{\mathsf{T}}\xspace}

\newcommand*{\tv}{\textsf{TV}\xspace}
\newcommand*{\te}{\textsf{TE}\xspace}

\newcommand*{\tav}{\ensuremath{\mathsf{TA^{V}}}\xspace}
\newcommand*{\tae}{\ensuremath{\mathsf{TA^{E}}}\xspace}
\newcommand*{\tga}{TGA\xspace}
\newcommand*{\tra}{TRA\xspace}

\newcommand*{\slice}[2]{\ensuremath{\tau_{#1}(#2)}}
\newcommand*{\subv}[2]{\ensuremath{\textsf{sub}^T_v(#1,#2)}}
\newcommand*{\sube}[2]{\ensuremath{\textsf{sub}^T_e(#1,#2)}}
\newcommand*{\vmap}[1]{\ensuremath{\textsf{map}^T_v(#1)}}
\newcommand*{\emap}[1]{\ensuremath{\textsf{map}^T_e(#1)}}

\newcommand*{\coal}[1]{\ensuremath{\mathcal{C}(#1)}}

\newcommand*{\constr}[2]{\ensuremath{\mathcal{K}(#1,#2)}}
\newcommand*{\resolve}[2]{\ensuremath{\mathcal{R}(#1,#2)}}
\newcommand*{\wsplit}[3]{\ensuremath{\mathcal{S}(#1,#2,#3)}}

\newcommand{\insql}[1]{\textsf{#1}}

\def\ojoin{\setbox0=\hbox{$\Join$}%
\rule[0.2ex]{.25em}{.4pt}\llap{\rule[1.5ex]{.25em}{.4pt}}}

\def\rightouterjoin{\mathbin{\Join\mkern-5.8mu\ojoin}}
\def\fullouterjoin{\mathbin{\ojoin\mkern-5.8mu\Join\mkern-6mu\ojoin}} 

\newcolumntype{L}[1]{>{\raggedright\let\newline\\\arraybackslash\hspace{0pt}}p{#1}}

\usepackage{etoolbox}
\makeatletter
\patchcmd{\maketitle}{\@copyrightspace}{}{}{}
\makeatother

\begin{document}

\title{Querying Evolving Graphs with Portal}

\numberofauthors{2}
\author{
  \alignauthor Vera Zaychik Moffitt\\
  \affaddr{Drexel University}\\
  \email{zaychik@drexel.edu}  \and
  \alignauthor Julia Stoyanovich\\
  \affaddr{Drexel University}\\
  \email{stoyanovich@drexel.edu}  \\
}

\maketitle

\thispagestyle{empty}

\begin{abstract}

Graphs are used to represent a plethora of phenomena, from the Web and
social networks, to biological pathways, to semantic knowledge
bases. Arguably the most interesting and important questions one can
ask about graphs have to do with their evolution. Which Web pages are
showing an increasing popularity trend? How does influence propagate
in social networks? How does knowledge evolve?  

This paper proposes a logical model of an evolving graph called a \tg,
which captures evolution of graph topology and of its vertex and edge
attributes.  We present a compositional temporal graph algebra \tga,
and show a reduction of \tga to temporal relational algebra with
graph-specific primitives.  We formally study the properties of \tga,
and also show that it is sufficient to concisely express a wide range
of common use cases.  We describe an implementation of our model and
algebra in \ql, built on top of Apache Spark / GraphX.  We conduct
extensive experiments on real datasets, and show that \ql scales.

\end{abstract}

\section{Introduction}
\label{sec:intro}

The importance of networks in scientific and commercial domains cannot
be overstated.  Considerable research and engineering effort is being
devoted to developing effective and efficient graph representations
and analytics.  Efficient graph abstractions and analytics for {\em
  static graphs} are available to researchers and practitioners in
scope of open source platforms such as Apache Giraph, Apache Spark /
GraphX~\cite{DBLP:conf/osdi/GonzalezXDCFS14} and GraphLab /
PowerGraph~\cite{DBLP:conf/osdi/GonzalezLGBG12}.

Arguably the most interesting and important questions one can ask
about networks have to do with their evolution, rather than with their
static state.  Analysis of {\em evolving graphs} has been receiving
increasing
attention~\cite{DBLP:journals/csur/AggarwalS14,Chan2008,Kan2009,Miao2015,Ren2011,Semertzidis2015}.
Yet, despite much recent activity, and despite increased variety and
availability of evolving graph data, systematic support for scalable
querying and analysis of evolving graphs still lacks.  This support is
urgently needed, due first and foremost to the scalability and
efficiency challenges inherent in evolving graph analysis, but also to
considerations of usability and ease of dissemination.  The goal of
our work is to fill this gap.  In this paper, we present an algebraic
query language called \tg algebra, or \tga, and its scalable
implementation in \ql, an open-source distributed framework on top of
Apache Spark.

Our goal in developing \tga is to give users an ability to concisely
express a wide range of common analysis tasks over evolving graphs,
while at the same time preserving inter-operability with temporal
relational algebra (\tra).  Implementing (non-temporal) graph querying
and analytics in an RDBMS has been receiving renewed
attention~\cite{DBLP:conf/sigmod/AbergerTOR16,DBLP:conf/sigmod/SunFSKHX15,DBLP:journals/pvldb/Xirogiannopoulos15},
and our work is in-line with this trend. Our data model is based on
the temporal relational model, and our algebra corresponds to temporal
relational algebra, but is designed specifically for evolving graphs.

We represent graph evolution, including changes in topology and in
attribute values of vertices and edges, using the \tg abstraction ---
a collection of temporal SQL relations with appropriate integrity
constraints.  An example of a \tg is given in Figure~\ref{fig:tg_ve},
where we show evolution of a co-authorship network.

\tga can be viewed as \tra for graphs, and so does not support general
recursion or transitive closure computation. (Although, as we will see
in Section~\ref{sec:analytics}, Pregel-style graph analytics such as
PageRank are supported as an extension.)  For this reason, we also do
not support regular path queries (RPQ) or the more general path query
classes (CRPQ and ECRPQ).  Extending our formalism with recursion and
path queries is non-trivial, and we leave this to future work.

Rather than focusing on path computation and graph traversal, we
stress the tasks that perform whole-graph analysis over time.  Several
such tasks are described next.  Additional examples can be found in
SocialScope~\cite{Amer-Yahia2009} --- a closed non-temporal algebra
for multigraphs that is motivated by information discovery in social
content sites.  It is not difficult to show that the graph (rather
than multigraph) versions of all SocialScope operations can be
expressed, and augmented with the temporal dimension, in \tga.

\subsection{Use cases and algebra by example}
\label{sec:cases}

An interaction network is one typical kind of an evolving graph.  It
represents people as vertices, and interactions between them such as
messages, conversations and endorsements, as edges.  Information
describing people and their interactions is represented by vertex and
edge attributes.  One easily accessible interaction network is the
wiki-talk dataset (\url{http://dx.doi.org/10.5281/zenodo.49561}),
containing messaging events among Wikipedia contributors over a
13-year period.  Information available about the users includes their
username, group membership, and the number of Wikipedia edits they
made.  Messaging events occur when users post on each other's talk
pages.

We now present common analysis tasks that motivate the operators of
our algebra, \tga. 

{\bf Vertex influence over time.} In an interaction graph, vertex
centrality is a measure of how important or influential people are.
Over a dozen different centrality measures exist, providing indicators
of how much information ``flows'' through the vertex or how the vertex
contributes to the overall cohesiveness of the network.  Vertex
importance fluctuates over time.  To see whether the wiki-talk graph
has high-importance vertices, and how stable vertex importance is over
time during a particular period of interest, we can look at a subset
of the graph that corresponds to the period of interest, compute an
importance measure, such as in-degree, for each vertex and for each
point in time, and finally calculate the coefficient of variation per
vertex.

{\bf Question:} What are the high-influence nodes over the past 5
years, and is their influence persistent over time?

\begin{enumerate}[noitemsep,itemindent=\dimexpr\labelwidth+\labelsep\relax,leftmargin=0pt]
\item Select a subset of the data representing the 5 years of
  interest, using a common temporal operator slice($\tau$):\\
$\ttt_1 = \slice{[2010,2015)}{wikitalk}$

\item Compute in-degree (prominence) of each vertex during each time
  point.  This is an example of the {\em aggregation} operation, a
  common operation on non-temporal graphs, as defined by the taxonomy
  of Wood~\cite{Wood2012}.  Aggregation computes a value for each
  vertex based on its neighbors.  SocialScope~\cite{Amer-Yahia2009} is
  one of the languages that proposes an aggregation operation and
  demonstrates its many uses.  We introduce a temporal version of
  aggregation (listed here with default arguments omitted for
  readability):\\
$\ttt_2 = \insql{agg}^T(\mathsf{dir=right},\mathsf{f_m=1},\mathsf{f_a=count},\mathsf{pname=deg},\ttt_1)$

\item Aggregate degree information per vertex across the timespan
  of $\ttt_2$, collecting values into a map.  This is an example of
  aggregation based on temporal window, which we implement with the
  temporal node creation operator:\\
$\ttt_3 =\insql{node}^T_w(\mathsf{w=lifetime},\mathsf{f_v=\{map(deg)\}},\ttt_2)$

\item Transform the attributes of each vertex to compute the
  coefficient of variation from the map of degree values, using the
  temporal vertex-map operator:\\
$\ttt_4 = \vmap{\mathsf{f_v=stdev(deg)/mean(deg)*100}, \ttt_3}$
\end{enumerate}

{\bf Graph centrality over time.} Graph centrality is a popular
measure that is used to evaluate how connected or centralized the
community is.  This measure can be computed by aggregating in-degree
values of graph vertices and may change as communication patterns
evolve, or as high influencers appear or disappear. In sparse
interaction graphs there is an additional question of temporal
resolution to consider: if two people communicated on May 16, 2010,
how long do we consider them to be connected?  We now show how graph
centrality can be computed over time, with control for temporal
resolution.

{\bf Question:} How has graph centrality changed over time?

\begin{enumerate}[noitemsep,,itemindent=\dimexpr\labelwidth+\labelsep\relax,leftmargin=0pt]
\item Compute a temporally aggregated view of the graph into 2-months
  windows.  Each window will include vertices and edges that
  communicate frequently: a vertex and an edge are each present during
  a 2-month window if they exist in every snapshot during that period.
  We use the window-based node creation operation.\\
$\ttt_1 = \insql{node}^T_w(\mathsf{w=2~mon},\mathsf{q_v=always},\mathsf{q_e=always},wikitalk)$

\item Compute in-degree of each vertex:\\
$\ttt_2 = \insql{agg}^T(\mathsf{dir=right},\mathsf{f_m=1},\mathsf{f_a=count},\mathsf{pname=deg},\ttt_1)$

\item Create a new graph, in which all vertices that are present at a
  given time point (snapshot) are grouped into a single vertex.
  Accumulate maximum, sum and count of the values of $deg$ as
  properties at that vertex.  We implement this with the
  attribute-based node creation operation.  Creating a single vertex
  to represent the whole graph is one use of node creation.  We will
  show that node creation is useful for other types of analysis. $\ttt_3 =$\\
  $\insql{node}^T_a(\mathsf{g=1,f_v=\{max(deg),~sum(deg),~count(deg)\}},\ttt_2)$

\item Compute degree centrality at each time point.\\
$\ttt_4 = \vmap{\mathsf{f_v=(max*cnt-sum)/(cnt^2-3*cnt+2)}, \ttt_3}$

\end{enumerate}

{\bf Communities over time.} Interaction networks are sparse because
edges are so short-lived.  As part of exploratory analysis, we can
consider the network at different temporal resolutions, run a
community detection algorithm, e.g., compute the connected components
of the network, and then consider the number of and size of connected
components.

{\bf Question:} In a sparse communication network, on what time scale
can we detect communities?

\begin{enumerate}[noitemsep,itemindent=\dimexpr\labelwidth+\labelsep\relax,leftmargin=0pt]
\item Aggregate the graph into 6-month windows.\\
$\ttt_1 = \insql{node}^T_w(\mathsf{w=6~mon},\mathsf{q_v=always},\mathsf{q_e=always},wikitalk)$

\item Compute connected components at each time point.  This is an
  example of a Pregel-style analytic invocation over an evolving
  graph. $\ttt_2 = \insql{pregel}^T_{cc} (\mathsf{pname=comp}, \ttt_1)$

\item Generate a new graph, in which a vertex corresponds to a
  connected component, and compute the size of the connected
  component. $\ttt_3 =
  \insql{node}^T_a(\mathsf{g=comp,f_v=count(1)},\ttt_2)$

\item Filter out vertices that represent communities too small to be
  useful (e.g., of 1-2 people).  This is an example of vertex
  subgraph.  $\ttt_4 = \subv{\mathsf{v.a.count > 2}}{\ttt_3}$

\end{enumerate}

In Section~\ref{sec:algebra} we formally define the operators of our
graph algebra.  In Section~\ref{sec:exp} we return to these three use
cases.

\subsection{Contributions and roadmap}

 We propose a representation of an evolving graph, called a \tg, which
 captures the evolution of both graph topology and vertex and edge
 attributes (Section~\ref{sec:model}), and develop a compositional \tg
 algebra, \tga (Section~\ref{sec:algebra}).
We show a reduction from \tga to temporal relational algebra
  \tra, using a combination of standard operators and \tg-specific
  primitives, and present formal properties of \tga
  (Section~\ref{sec:formal}).
We present an implementation in scope of the \ql system, built
  on Apache Spark / GraphX~\cite{DBLP:conf/osdi/GonzalezXDCFS14}. \ql
  supports several access methods that correspond to different
  trade-offs in temporal and structural locality
  (Section~\ref{sec:sys}).
We conduct an extensive experimental evaluation with real datasets,
demonstrating that \ql scales (Section~\ref{sec:exp}).  We also
illustrate the usability throughout the paper, with a variety of
real-life analysis tasks that can be concisely expressed in \tga.

\section{Data Model}
\label{sec:model}

Following the SQL:2011
standard~\cite{DBLP:journals/sigmod/KulkarniM12}, a period (or
interval) $\bp = [s, e)$ represents a discrete set of time instances,
  starting from and including the start time $s$, continuing to but
  excluding the end time $e$.  Time instances contained within the
  period have limited precision, and the time domain has total order.

\begin{figure}[t!]
\centering
\includegraphics[width=3in]{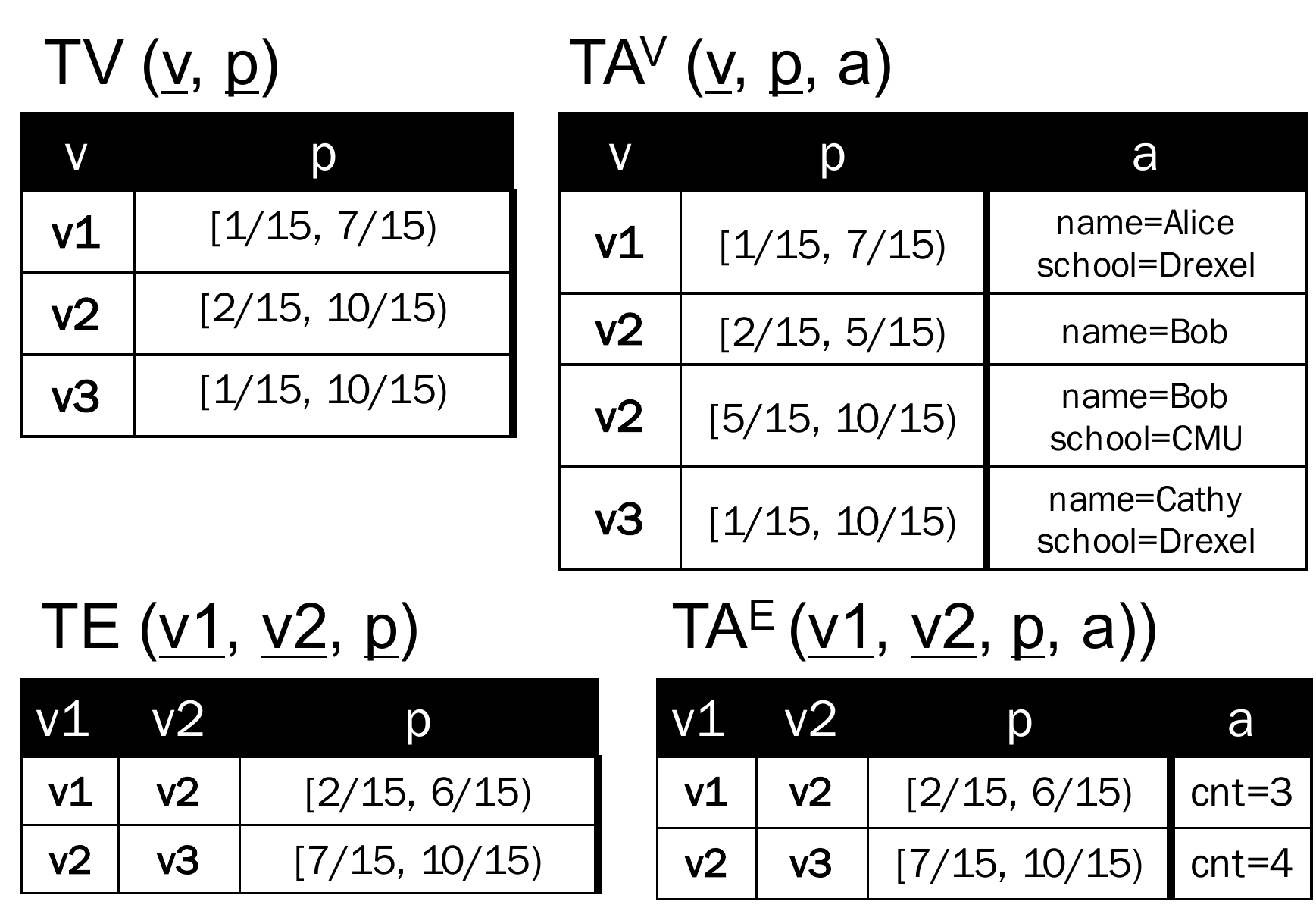}
\caption{\tg \insql{T1}.}
\label{fig:tg_ve}
\end{figure}

We now describe the logical representation of an evolving graph,
called a \tg.  A \tg represents a single graph, and models evolution
of its topology and of vertex and edge attributes.
Figure~\ref{fig:tg_ve} gives an example of a \tg that shows evolution
of a co-authorship network.

A \tg is represented with four temporal SQL
relations~\cite{DBLP:conf/vldb/BohlenSS96}, and uses point
semantics~\cite{DBLP:reference/db/Toman09}, associating a fact
(existence of a vertex or edge, and an assignment of a value to a
vertex or edge attribute) with a time point.  We use periods to
compactly represent their constituent time points.  This is a common
representation technique, which does not add expressive power to the
data model~\cite{DBLP:conf/ictl/Chomicki94}.

A snapshot of a temporal relation $R$, denoted $\tau^s_c(R)$ (``s''
stands for ``snapshot''), is the state of $R$ at time point $c$.

We use the property graph model~\cite{GraphDB} to represent vertex and
edge attributes: each vertex and edge during period $\bp$ is
associated with a (possibly empty) {\em set} of properties, and each
property is represented by a key-value pair.  Property values are not
restricted to be of atomic types, and may, e.g., be sets, maps or
tuples.

We now give a formal definition of a \tg.

\begin{definition}[TGraph]
A \tg is a pair $\tve=(\tv, \te)$. \tv is a valid-time temporal SQL
relation with schema $\tv(\underline{v}, \underline{\bp})$ that
associates a vertex with the time period during which it is
present. \te is a valid-time temporal SQL relation with schema
$\te(\underline{v_1}, \underline{v_2}, \underline{\bp})$, connecting
pairs of vertices from \tv.  
\tve optionally includes vertex and edge attribute relations
$\tav(\underline{v}, \underline{\bp}, a)$ and $\tae(\underline{v_1},
\underline{v_2}, \underline{\bp}, a)$.  
Relations of \tve must meet the following requirements:

\begin{description}[noitemsep]
\item [R1: Unique vertices/ edges] In every snapshot $\tau^s_c
  (\tav)$ and $\tau^s_c (\tae)$, where $c$ is a time point, a
  vertex/edge exists at most once.
\item [R2: Unique attribute values] In every snapshot $\tau^s_c
  (\tav)$ and $\tau^s_c (\tae)$, a vertex/edge is associated with at
  most one attribute (which is itself a set of key-value pairs
  representing properties).
\item [R3: Referential integrity] In every snapshot $\tau^s_c (\tve)$,
  foreign key constraints hold from $\tau^s_c (\te)$ (on both $v_1$
  and $v_2$) and $\tau^s_c (\tav)$ to $\tau^s_c (\tv)$, and from
  $\tau^s_c (\tae)$ to $\tau^s_c (\te)$.
\item [R4: Coalesced] Value-equivalent tuples in all relations of \tve
  with consecutive or overlapping time periods are merged.  
\end{description}
\label{def:tg}
\vspace{-0.5cm}
\end{definition}

Requirements {\bf R1, R2, R3} guarantee soundness of the \tg data
structure, ensuring that every snapshot of a \tg is a valid graph.
Requirement {\bf R4} avoids semantic ambiguity and ensures correctness
of algebraic operations in point-stamped temporal models such as
ours~\cite{DBLP:reference/db/JensenS09k}.

Graphs may be directed or undirected.  For undirected graphs we choose
a canonical representation of an edge, with $v_1 \leq v_2$ (self-loops
are allowed).  Because we use the source and destination vertex id
pair as the identifier for the edges, at most two edges can exist
between any two vertices (one in each direction) at any time point.
That is, we do not support multigraphs.

In the \tg representation of Definition~\ref{def:tg}, vertex and edges
attributes are stored as collections of properties.  That said,
Definition~\ref{def:tg} presents a logical data structure that admits
different physical representations, including, e.g., a columnar
representation (each property in a separate relation, supporting
different change rates), by a hash-based representation
of~\cite{DBLP:conf/sigmod/SunFSKHX15}, or in some other way. We leave
an experimental comparison of different physical representations of
vertex and edge attributes to future work.

Our choice to use attribute relations is in contrast to representing
vertex and edge attributes as part of \tv and \te.  The main reason is
to streamline the enforcement of referential integrity constraints of
Definition~\ref{def:tg}.  Consider again the example in
Figure~\ref{fig:tg_ve}.  If vertex attributes were stored as part of
\tv, then there would be two tuples for $v_2$ in this relation whose
validity periods overlap with that of edge $e(v_1, v_2)$ --- one for
each $[2/15, 5/15)$ and $[5/15, 7/15)$.  This would in turn require
    that $e(v_1, v_2)$ be mapped to two tuples in \tv as part of
    referential integrity checking on $v_2$.  Matching a tuple with a
    set of tuples in the referenced table, while supported by the
    SQL:2011 standard, is potentially inefficient, and we avoid it in
    our representation.
Another reason for storing \tav and \tae separately is that in many
cases we are interested in applying operations (e.g., analytics) only
to graph topology, and in that case \tv and \te are sufficient.

\section{Algebra}
\label{sec:algebra}
\setlength{\textfloatsep}{5pt}

\tg algebra, or \tga for short, is compositional: operators take a \tg
or a pair of \tgs as input, and output a \tg.  We specify the
semantics of \tga by showing a translation of each operator into a
sequence of temporal relational algebra (\tra) expressions (with
nesting, to accommodate non-1NF vertex/edge attributes).  Using this
translation one can implement \tga in a temporal DBMS, guaranteeing
snapshot reducibility and extended snapshot
reducibility~\cite{DBLP:reference/db/Bohlen092} --- two properties
that are appropriate for a point-based temporal data model.

\tra algebra extends relational algebra by specifying how operators
are applied to temporal relations such that snapshot reducibility
property is guaranteed.  Additionally, explicit references to time are
supported in operator predicates (extended snapshot reducibility), but
the time stamps are not manipulated by the user queries directly.

In Section~\ref{sec:algebra:integrity} we present the primitives that
are needed to enforce soundness of \tga.  Then, in
Sections~\ref{sec:algebra:unary} through~\ref{sec:algebra:ecreate}, we
present \tga operators.  Section~\ref{sec:analytics} presents an
extension of \tga to support Pregel-style analytics.

\subsection{Primitives and Soundness}
\label{sec:algebra:integrity}

\tga operators are translated into expressions in temporal relational
algebra (\tra).  Since TRA is applied to individual relations of \tve,
we must ensure that the combined state of these relations in the
result corresponds to a valid \tg, i.e., that the translation is
sound.  Recall from Definition~\ref{def:tg} that a valid \tg must
satisfy four requirements: {\bf R1}: Unique vertices and edges, {\bf
  R2}: Unique attribute values, {\bf R3}: Referential integrity, and
{\bf R4}: Coalesced.  We now describe four primitives that will ensure
soundness of \tga.

{\bf Coalesce.} To enforce requirements {\bf R1} and {\bf R4}, we
introduce the coalesce primitive $\coal{R}$, which merges adjacent or
overlapping time periods for value-equivalent tuples.  This operation
is similar to duplicate elimination in conventional databases, and has
been extensively studied in the
literature~\cite{DBLP:conf/vldb/BohlenSS96,DBLP:journals/sigmod/Zimanyi06}.
$\coal{R}$ is applied to individual relations of \tve, or to
intermediate results, following the application of operations that
uncoalesce.
The coalesce primitive can be implemented in relational
algebra~\cite{DBLP:conf/vldb/BohlenSS96}.  If a DBMS supports
automatic coalescing, this primitive is not necessary.

{\bf Resolve.} To enforce {\bf R2} we introduce the resolve primitive
$\resolve{f_1(k_1), \ldots, f_n(k_n)}{R}$, which is invoked by
operations that produce attribute relations with duplicates.  Resolve
computes a temporal group-by of the attribute relation $R$ by key
(e.g., by $v$ if $R$ represents vertex attributes).  It then computes
a bag-union of the properties occurring in each group, groups together
key-value pairs that correspond to the same property name $k_i$, and
aggregates values within each group using the specified aggregation
function $f_i$.  If no aggregation function is specified for a
particular property name, \insql{set} is used as the default.  For
example, if $R$ contains tuples $(v_1, [2/15,4/15),
  \{name=Ann,sal=100\})$ and $(v_1, [2/15, 3/15),
    \{name=Ann,sal=200\})$, the result of $\resolve{AVG(sal)}{R}$ will
    contain $(v_1, [2/15,3/15), \{name=Ann,sal=150\})$ and $(v_1,
      [3/15,4/15), \{name=Ann,sal=100\})$.  The resolve primitive can
        be implemented with temporal relational aggregation $\gamma^T$
        over unnested relations.

{\bf Constrain.} To enforce {\bf R3} we introduce the constrain
primitive $\constr{r}{s}$, which enforces referential integrity on
relation $\mathbf{r}$ with respect to relation $\mathbf{s}$.  For
example, this primitive is used to remove edges from \te for which one
or both vertices are absent from \tv, or restrict the validity period
of an edge to be within the validity periods of its vertices.

{\bf Split.} The final primitive $\wsplit{s}{w}{R}$ uncoalesces
relation $R$ in a particular way.  For each tuple $t \in R$ with time
period $p$, it emits a set of tuples, with the same values for
non-temporal attributes as in $t$, but with time periods split into
windows of width $w$ with respect to start time $s$.  For example,
$\wsplit{2/15}{3~\textsf{months}}{\te}$ for \insql{T1} in
Figure~\ref{fig:tg_ve} produces an uncoalesced relation with 2 tuples
for $(v_1, v_2)$, with periods $[2/15, 5/15)$ and $[5/15, 6/15)$, and
    2 tuples for $(v_2, v_3)$, with validity periods $[7/15, 8/15)$
      and $[8/15, 10/15)$.  This primitive will be necessary to
        express the temporal variant of node creation
        (Section~\ref{sec:algebra:ncreate}).  It will be used
        at an intermediate step in the computation, all final results
        will be coalesced as needed, enforcing {\bf R4}.

\subsection{Unary operators}
\label{sec:algebra:unary}

\begin{figure*}[t]
\begin{minipage}[b]{2.5in}
\includegraphics[width=2.5in]{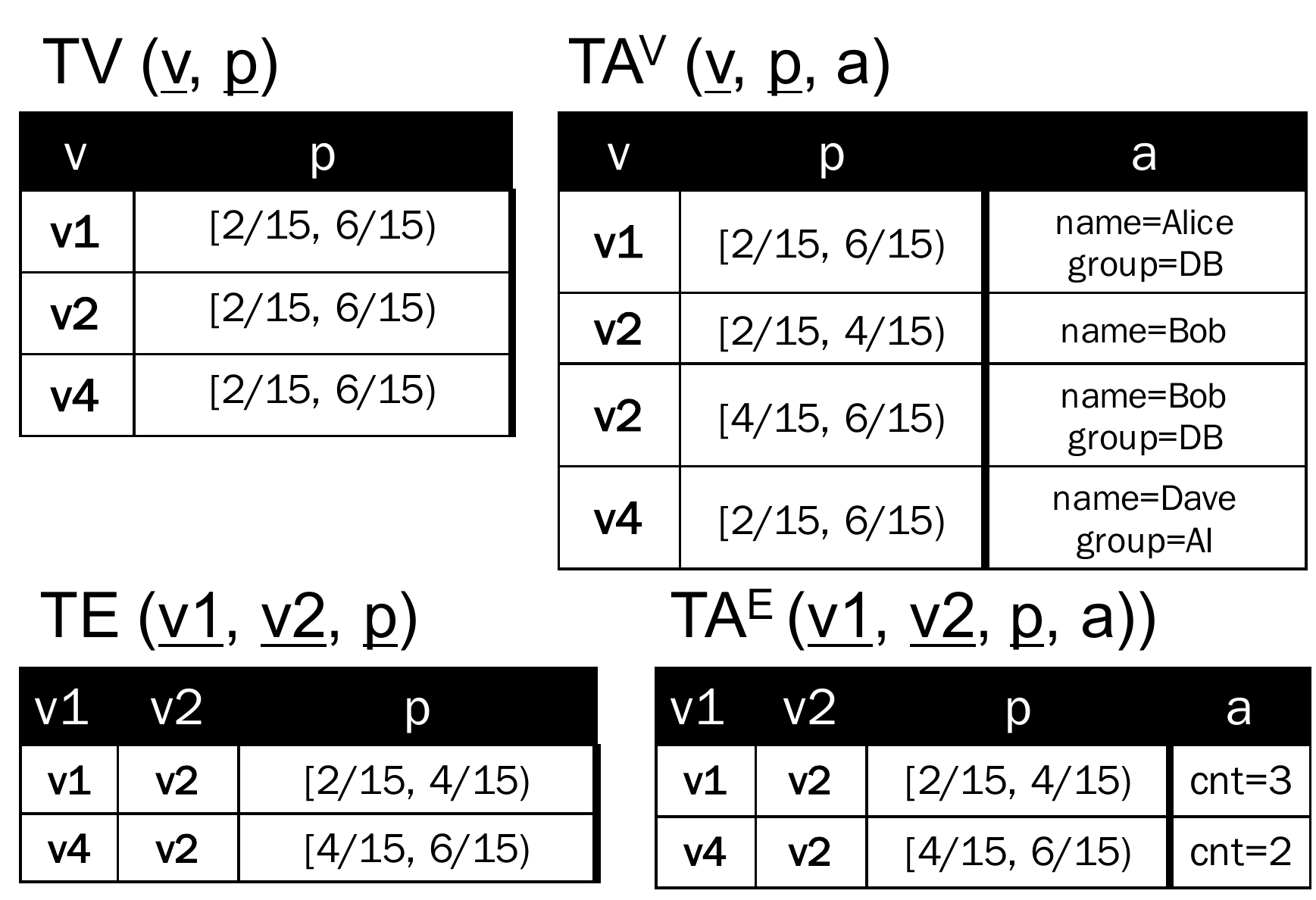}
\caption{T2.}
\label{fig:tg_t2}
\vspace{-0.1cm}
\end{minipage}
\begin{minipage}[b]{4.3in}
\includegraphics[width=4.3in]{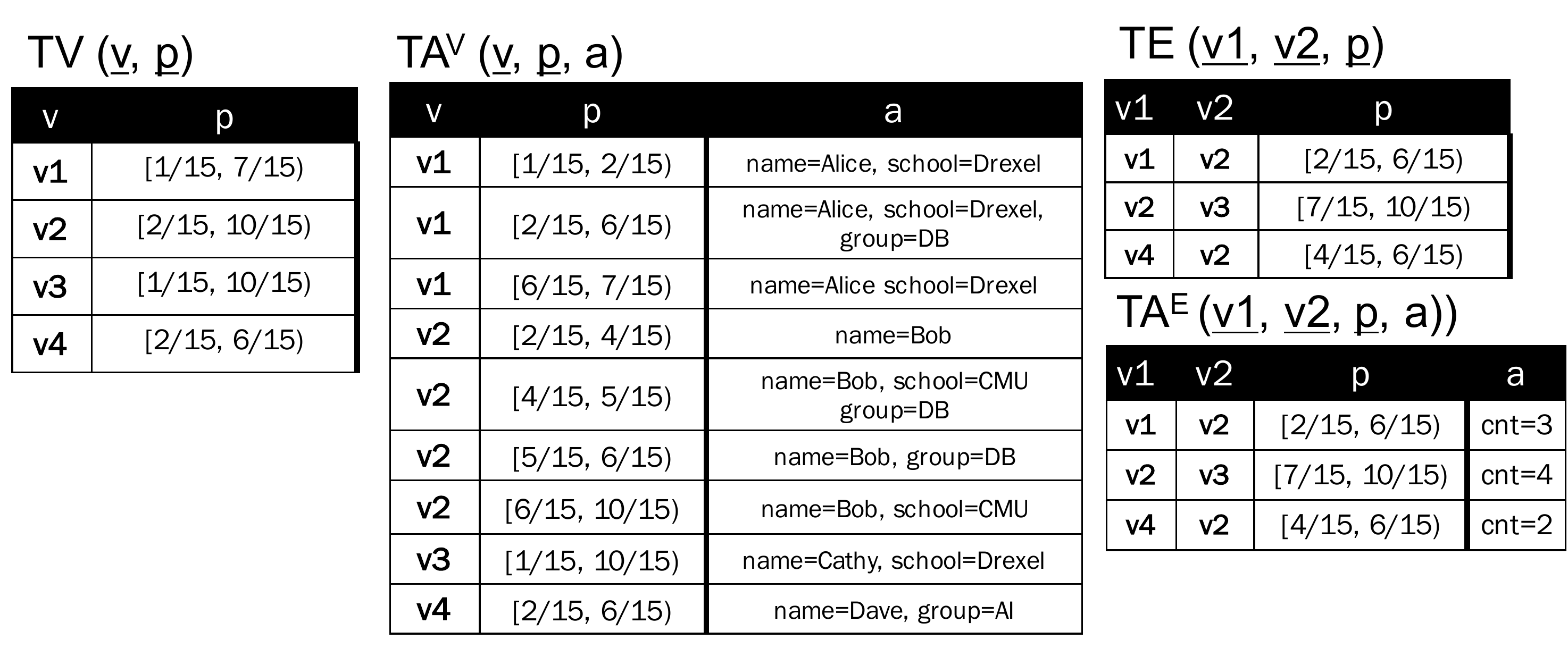}
\caption{$T1 \cup^T T2.$}
\label{fig:tg_union}
\vspace{-0.1cm}
\end{minipage}
\end{figure*}

{\bf Slice.} The slice operator, denoted $\slice{\bc}{\ttt}$, where
$\bc$ is a time interval, cuts a temporal slice from \ttt. The
resulting \tg will contain vertices and edges whose period $\bp$ has a
non-empty intersection with $\bc$.  We translate this \tga operator to
TRA statements over each constituent relation of \tve:
$\slice{\bc}{\tv}$ and similarly for \te, \tav and \tae.

{\bf Subgraph.} Temporal subgraph matching is a generalization of
subgraph matching for non-temporal graphs~\cite{Wood2012}.  This query
comes in two variants.

Temporal vertex-subgraph \subv{q^t_v}{\ttt} computes an induced
subgraph of \tve $\tve'(\tv', \te', \tav', \tae')$, with vertices
defined by the temporal conjunctive query (TCQ) $q^t_v$.  Note that
this is a subgraph query, and so $\tv' \subseteq^T \tv$.

Temporal edge-subgraph \sube{q^t_e}{\ttt} computes a subgraph of \ttt
$\tve'(\tv', \te', \tav', \tae')$ in which edges are defined by TCQ
$q^t_e$.  Since this is a subgraph query, $\te' \subseteq^T \te$.

Queries $q^t_v$ and $q^t_e$ may use any of the constituent relations
of \tve, and may explicitly reference temporal information, and so
require all input relations to be
coalesced~\cite{DBLP:reference/db/Bohlen09}.

Following the computation of $\tv' = q^t_v(\tv)$, \subv{q^t_v}{\ttt}
must invoke $\coal{\tv'}$ to enforce {\bf R1} and {\bf R4}; and
$\constr{\te'}{\tv'}$, $\constr{\tav'}{\tv'}$, $\constr{\tae'}{\te'}$
to enforce {\bf R3}.
Following the computation $\te' = q^t_e(\te)$, \sube{q^t_e}{\ttt} must
invoke $\coal{\te'}$ to enforce {\bf R1} and {\bf R4}; and
$\constr{\tae'}{\te'}$ to enforce {\bf R3}.  

{\bf Map.}  Temporal vertex-map and edge-map apply user-defined map
functions $f_v$ and $f_e$ to vertex or edge attributes.  Temporal
vertex-map $\vmap{f_v, \tve}$ outputs $\tve'$ in which $\tv'=\tv$,
$\te'=\te$, $\tae' = \tae$, and $\tav' =
\pi^T_{v,f_v(a)}\tav$. Temporal edge-map $\emap{f_e, \tve}$ is defined
analogously.

While $f_v$ and $f_e$ are arbitrary user-specified functions, there
are some common cases.  Map may specify the set of properties to
project out or retain, it may aggregate (e.g., \insql{COUNT}) or
deduplicate values of a collection property, or flatten a nested
value.
To produce a valid \tg, $\vmap{f_v, \tve}$ must invoke $\coal{\tav'}$
and $\emap{f_e, \tve}$ must invoke $\coal{\tae'}$.

\subsection{Aggregation}
\label{sec:algebra:agg}

Aggregation is a common graph operation that computes the value of a
vertex property $pname$ based on information available at the vertex
itself, at the edges associated with the vertex, and at its immediate
neighbors.  Aggregation can be used to compute simple properties such
as in-degree of a vertex, or more complex ones such as the set of
countries in which the friends of $v$ live.

It is convenient to think of aggregation as operating over a temporal
view $L(v_1,v_2,v_1.a,v_2.a,e.a,p)$, where $v_1$ refers to the vertex
for which the new property is being computed, $v_2$ refers to the
vertex from which information is gathered, $v_1.a$, $v_2.a$ and $e.a$
are attributes of the vertices and of the edge, and $p$ is the
associated time period.  $L$ is computed with a temporal join of \te
with two copies of \tv, one for each side of the edge, and with $\tav$
and $\tae$ outer-joined with the corresponding relations.  Outer-joins
are needed because a vertex / edge is not required to specify an
attribute.

When \tve represents a directed graph, direction of the edge can be
accounted for in the way the join is set up (e.g., mapping $v_2$ in
\te to $v_1$ in $L$ if the goal is to aggregate information on
incoming edges).  When \tve represents an undirected graph (recall
that we choose a canonical representation of an edge, with $v_1 \leq
v_2$), or when direction of the edge is unimportant, $L$ can be
computed from $\te(v_1,v_2,p) \cup^T \te(v_2,v_1,p)$ rather than from
$\te$.

Aggregation is denoted $\insql{agg}^T(dir,cond,f_m,f_a,pname,\ttt)$,
where $dir$ is the direction of the edge (one of 'right', 'left' or
'both') that determines how $L$ is computed, $cond$ is a predicate
over $L$, $f_m$ is a map function that emits a value for each tuple in
the result of $\sigma^T_{cond}(L)$ (e.g., 1 for computing degree of
$v_1$, or $v_2.a.country$ for computing the set of countries in which
the friends of $v_1$ live).  Finally, $f_a$ is the function that
aggregates values computed by $f_m$, and $pname$ is the name of the
property to which the computed value is assigned.  Putting everything
together, and omitting the computation of $L$ for clarity: we compute
a temporal relation $R = \coal{v_1 \gamma^T_{f_a} (\pi^T_{v_1,f_m}
  (\sigma^T_{cond} L))}$. (Here, $\gamma^T$ is the temporal version of
relational aggregation, and $v_1$ is the grouping attribute.)  We then
compute an outer join of \tav with $R$, and invoke the resolve
primitive to reconcile the newly-computed property stored in $R.a$
with $\tav.a$: $\tav' = \resolve{\insql{set}(pname)}{\tav
  \rightouterjoin^T_{v=v_1} R}$. Note the use of the resolve primitive
at the last step.  Although there are no duplicates in the result of
the outer join of \tav and $R$, since $R$ is temporally coalesced and
the join is by key, resolve is needed to compute a bag-union of
properties in $R.a$ and $\tav.a$, and to aggregate the values
corresponding to $pname$ (in case $pname$ already occurred as a
property in $\tav.a$).

We support various aggregation functions $f_a$, including the standard
\{ \insql{count} | \insql{min} | \insql{max} | \insql{sum} \}, which
have their customary meaning.  We also support \{ \insql{any} |
\insql{first} | \insql{last} | \insql{set} | \insql{list} \}, which
are possible to compute because properties being reduced have temporal
information.  \insql{first} and \insql{last} refer to the value of a
property with the earliest/latest timestamp, while \insql{set} and
\insql{list} associate a key with a collection of values.

\subsection{Binary set operators}
\label{sec:algebra:binary}

We support temporal versions of the three binary set operators
intersection ($\cap^T$), union ($\cup^T$), and difference
($\setminus^T$).

To compute $\insql{T1} \cup^T \insql{T2}$, we set $\tv' = \tv_1 \cup^T
\tv_2$ and $\te' = \te_1 \cup^T \te_2$.  Next, we compute $\tav' =
\resolve{f_v}{\tav_1 \fullouterjoin^T_{v} \tav_2}$ and $\tae' =
\resolve{f_e}{\tae_1 \fullouterjoin^T_{v1,v2} \tae_2}$. 

Consider \insql{T1} in Figure~\ref{fig:tg_ve} and \insql{T2} in
Figure~\ref{fig:tg_t2}.  Figure~\ref{fig:tg_union} illustrates
\insql{T1} $\cup^T$ \insql{T2}.  According to the definition of
$\cup^T$, periods are split to coincide for any group, and thus the
attribute values for e.g., $v_1$ have three distinct tuples.

To compute $\insql{T1} \cap^T \insql{T2}$, we set $\tv' = \tv_1 \cap^T
\tv_2$ and $\te' = \te_1 \cap^T \te_2$.  Next, we compute $\tav' =
\constr{\resolve{f_v}{\tav_1 \fullouterjoin^T_{v} \tav_2}}{\tv'}$ and
$\tae' = \constr{\resolve{f_e}{\tae_1 \fullouterjoin^T_{v1,v2}
    \tae_2}}{\te'}$.

As an example, when applying \insql{T1} $\cap^T$ \insql{T2}, only the
vertices and edges present in both \tgs are produced, thus eliminating
$v_3$ and $v_4$.  Period $[2/15, 4/15)$ for $v_2$ is computed as a
  result of the join of $[2/15, 5/15)$ in \insql{T1} and [$2/15,
      4/15)$ in \insql{T2}.  See Figure~\ref{fig:tg_inter} in
      Appendix~\ref{sec:app:examples} for full result.

To compute $\insql{T1} \setminus^T \insql{T2}$, we set $\tv' = \tv_1
\setminus^T \tv_2$ and $\te' = \te_1 \setminus^T \te_2$.  Next, we
compute $\tav' = \constr{{\tav}_1}{\tv'}$ and $\tae' =
\constr{{\tae}_1}{\te'}$.

To continue the example above, the result of \insql{T1} $\setminus^T$
\insql{T2} includes vertex v1 before 2/15 and after 6/15, splitting
one v1 tuple in \tv of T1 into two temporally-disjoint tuples in the
result.  See Figure~\ref{fig:tg_diff} in
Appendix~\ref{sec:app:examples} for full result.

Note that both $\cap^T$ and $\cup^T$ require that resolve be invoked,
to reconcile the vertex/edge attributes associated with vertices/edges
in the temporal intersection of the inputs.

\subsection{Node creation}
\label{sec:algebra:ncreate}

\begin{figure}[b]
\includegraphics[width=3in]{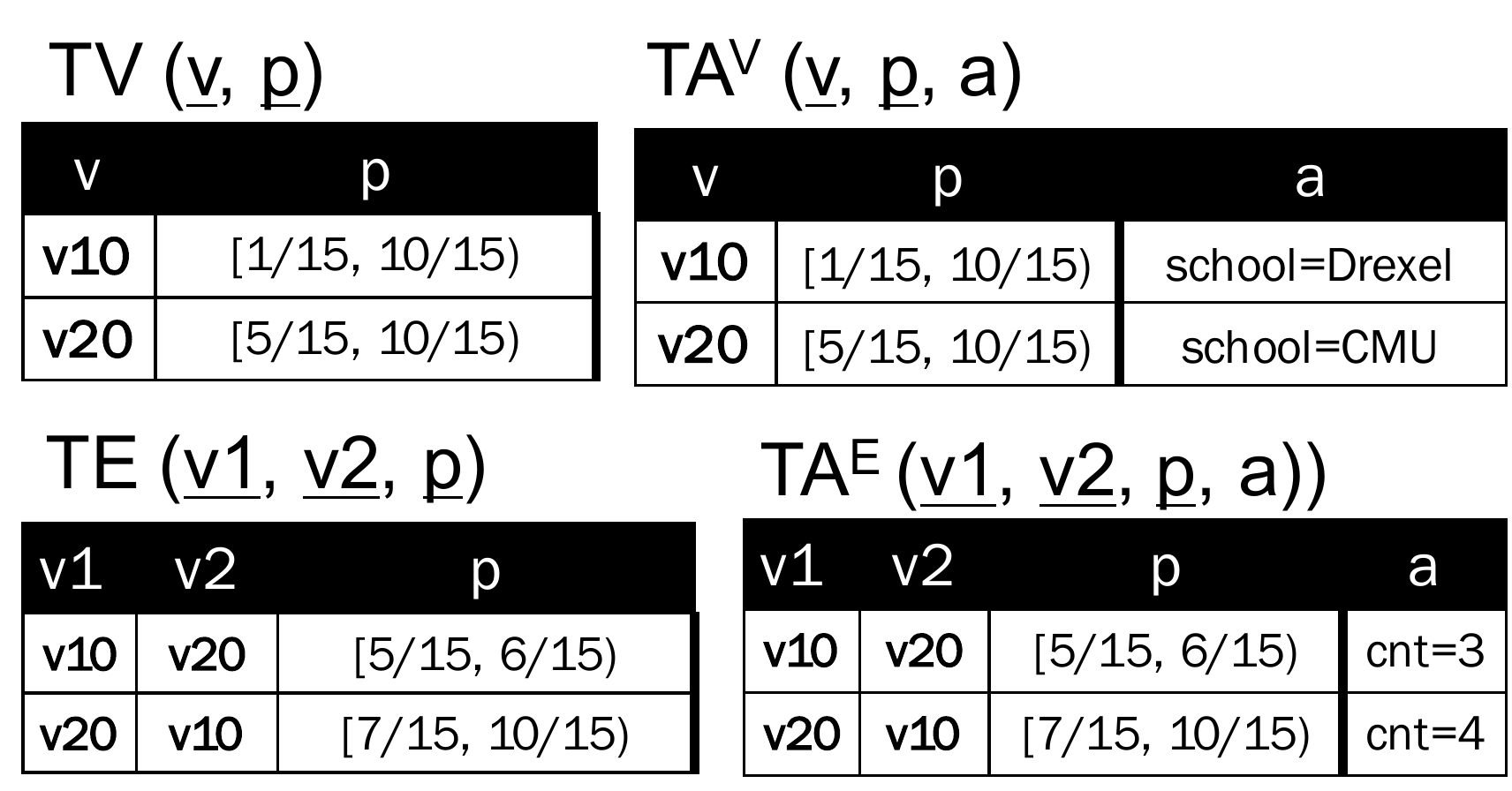}
\caption{\insql{node}$^T_a(school, set(school), max(cnt).)$}
\label{fig:tg_agg3}
\vspace{-0.1cm}
\end{figure}

The node creation operator enables the user to analyze an evolving
graph at different levels of granularity.  This operator comes in two
variants --- based on vertex attributes or based on temporal window.

{\bf Attribute-based node creation} is denoted\\
$\insql{node}^T_a(g_1,\ldots,g_i,f_{v1}(k_1),\ldots,f_{vn}(k_n),f_{e1}(l_1),\ldots,f_{em}(l_m),\tve)$,
where $g_1,\ldots,g_i$ are the grouping attributes, and each
$f_{vj}(k_j)$ ($f_{ej}(l_j)$) specifies an aggregation function
$f_{vj}$ (resp. $f_{ej}$) to be applied to a vertex property $k_j$
(resp. edge property $l_j$).  This operation allows the user to
generate a \tg in which vertices correspond to disjoint groups of
vertices in the input that agree on the values of all grouping
attributes.  For example, $\insql{node}^T_a(school,\tve)$ will compute
a vertex for each value of $\tav.a.school$.  Vertices that do not
specify a value for one or several grouping attributes at a given
time, will not contribute to the result for the corresponding
snapshot.

To compute $\tve' =
\insql{node}^T_a(g_1,\ldots,g_i,f_{v1}(k_1),\ldots,f_{vn}(k_n),$
\\ $f_{e1}(l_1),\ldots,f_{em}(l_m),\tve)$, we execute $L =
g_1,\ldots,g_i \gamma^T (\tv \rightouterjoin^T_v \tav)$, computing an
intermediate temporal relation for each group.  Next, we generate the
new vertex relation $\tv'$, by generating an id for each group with a
Skolem function: $\tv' = \sigma^T_{skolem(g_1,\ldots,g_i)}L$.

Then, $\tav' = \resolve{f_{v1}(k_1), \ldots,
  f_{vn}(k_n)}{\sigma^T_{skolem(g_1,\ldots,g_i), a} L}$.  Note the use
of the resolve primitive to reconcile attribute values within a group.

Vertices of the input are partitioned on their values of the grouping
attributes.  Partitioning of the vertices also induces a partitioning
of the edges. To compute the new edges $\te'$, we generate a temporal
conjunctive query that computes $E(v_1,v_2,p,v_1',v_2')$, where $v_1'$
and $v_2'$ are the identifiers of the vertices in \tv' to which $v_1$
and $v_2$ are mapped.  Finally, we compute $\te' =
\coal{\pi^T_{v_1',v_2'}E}$ and\\ $\tae' =
\resolve{f_{e1}(l_1),\ldots,f_{vm}(l_m)}{v_1',v_2' \gamma^T (E
  \bowtie^T_{v_1,v_2} \tae)}$.

Figure~\ref{fig:tg_agg3} illustrates attribute-based node creation
over \insql{T1} in our running example, with \insql{set(school)}
aggregation function for vertices and and \insql{max(cnt)} for edges.
Vertices $v_1$ and $v_3$ create a single new vertex $v_{10}$,
representing Drexel.

{\bf Window-based node creation} is denoted\\
$\insql{node}^T_w(w,q_v,q_e,f_{v1}(k_1),\ldots,f_{vn}(k_n),f_{e1}(l_1),\ldots,f_{em}(l_m),\tve)$,
where $w$ is the window specification, $q_v$ and $q_e$ are vertex and
edge quantifiers, and each $f_{vj}(k_j)$ ($f_{ej}(l_j)$) specifies an
aggregation function $f_{vj}$ (resp. $f_{ej}$) to be applied to a
vertex property $k_j$ (resp. edge property $l_j$).  This operation
corresponds to moving window temporal aggregation, and is inspired by
the stream aggregation work of~\cite{Li2005} and by generalized
quantifiers of~\cite{Hsu1995}, both adopted to graphs.

Window specification $w$ is of the form $n~\{unit|\insql{changes}\}$,
where $n$ is an integer, and $unit$ is a time unit, e.g., $10~min$,
$3~years$, or any multiple of the usual time units.  When $w$ is the
form $n~\insql{changes}$, it defines the window by the number of
changes that occurred in \tve (affecting any of its constituent
relations). Window boundaries are computed left-to-right, i.e., from
least to most recent.  

Vertex and edge quantifiers $q_v$ and $q_e$ are of the form \{
\insql{all} | \insql{most} | \insql{at least} $n$ | \insql{exists} \},
where $n$ is a decimal representing the percentage of the time during
which a vertex or an edge existed, relative to the duration of the
window (\insql{exists} is the default).  Quantifiers are useful for
observing different kinds of temporal evolution, e.g., to observe only
strong connections over a volatile evolving graph, we may want to only
include vertices that span the entire window ($q_v=\insql{all}$), and
edges that span a large portion of the window ($q_e=\insql{most}$).
 
For both kinds of window specification (by unit or by number of
changes), we must (1) compute a mapping from a tuple in a temporal
relation to one or multiple windows, and (2) aggregate over each
window.  The $s$ parameter for the split primitive is the smallest
start date across \tve.

To compute $\tv'$, we apply split to \tv, group by vid, select only
those vertices that meet the quantification, and finally coalesce:
$\tv' = \coal{\sigma^T_{q_V} (v \gamma^T_{\cup
    p}(\wsplit{s}{w}{\tv}))}$.  Similarly for $\te'$.  To compute
attribute relations we split, resolve with the aggregation functions,
and constrain: $\tav' =
\constr{\resolve{f_{v1}(k_1),\ldots,f_{vn}(k_n)}{\wsplit{s}{w}{\tav}}}{\tv'}$,
similarly for $\tae'$.

\begin{figure*}[t]
\begin{subfigure}[b]{0.5\textwidth}
\includegraphics[width=3.2in]{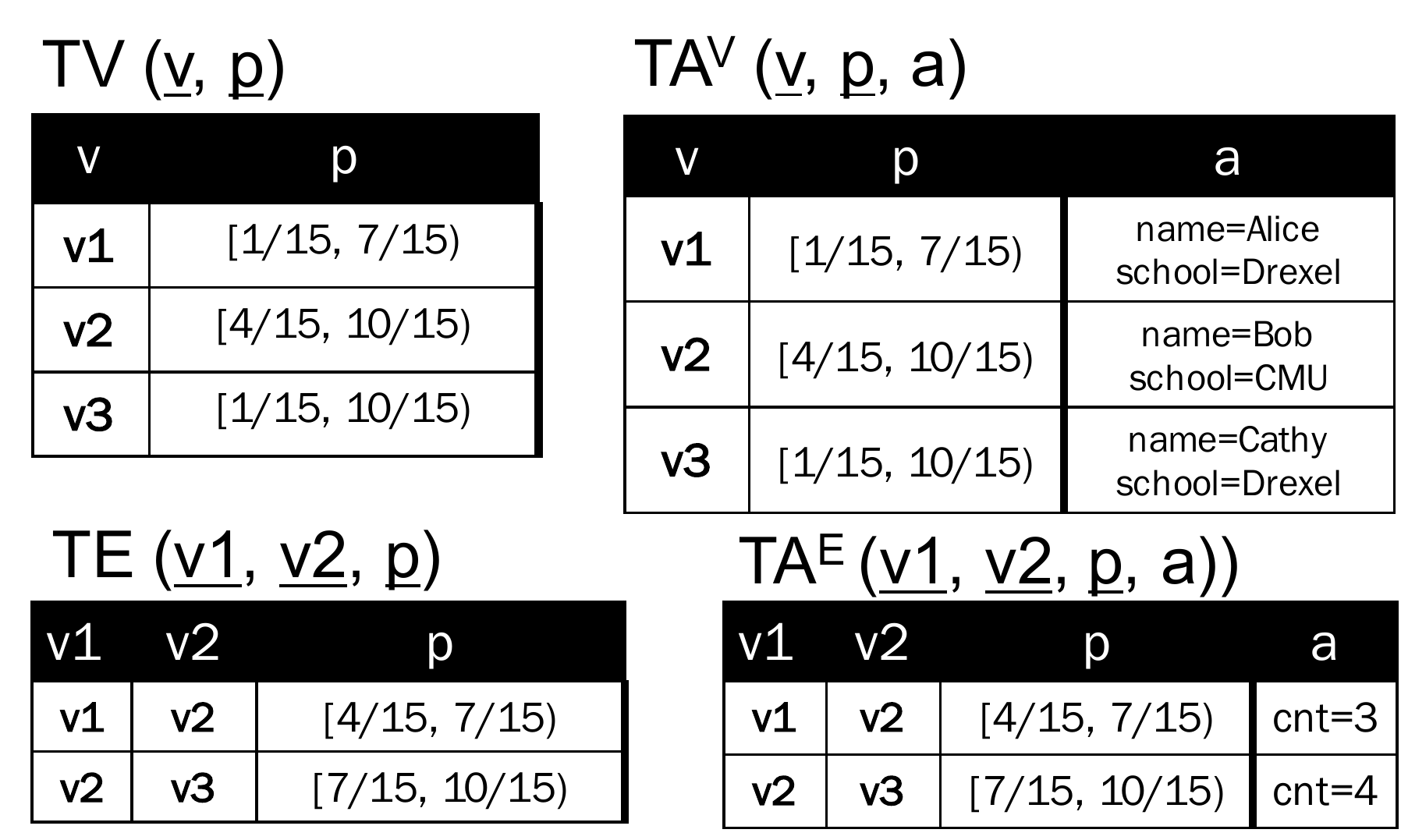}
\caption{By time: $w=3~\textsf{months}$.}
\label{fig:tg_agg1}
\end{subfigure}
\begin{subfigure}[b]{0.5\textwidth}
\includegraphics[width=3.2in]{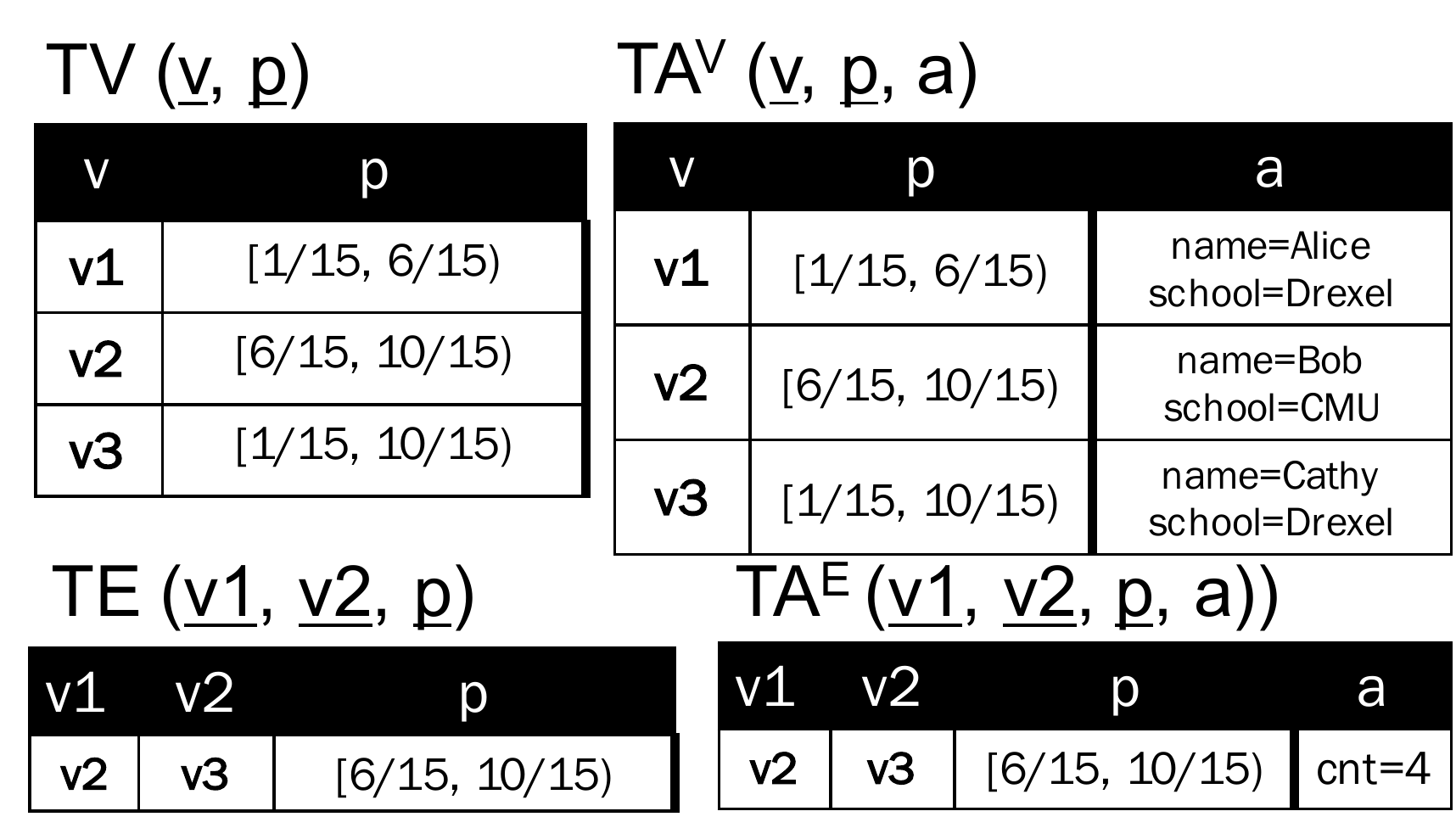}
\caption{By change: $w=3~\textsf{changes}$.}
\label{fig:tg_agg2}
\end{subfigure}
\caption[]{Node creation,
  $\insql{node}^T_w(\mathsf{q_v=always},\mathsf{q_e=exists,f_v=\{first(name),~first(school)\}},\ttt)$.}
\label{fig:tg_agg}
\vspace{-0.4cm}
\end{figure*}

Figure~\ref{fig:tg_agg1} illustrates window-based node creation by
time ($w=3~\textsf{months}$), and Figure~\ref{fig:tg_agg2} --- by
change ($w=3~\textsf{changes}$).  Both are applied to \insql{T1} in
our running example with \insql{all} quantifier for vertices and
\insql{exists} for edges, and \insql{first} aggregation function for
vertex and edge properties.  $v_2$ is present in the result in
Figure~\ref{fig:tg_agg1} starting at $4/15$ because it did not exist
for the entirety of the first window, while in
Figure~\ref{fig:tg_agg2} it is produced starting $6/15$.

\subsection{Edge creation}
\label{sec:algebra:ecreate}

Edge creation
$\insql{edge}^T(q,f_1(k_1),\ldots,f_n(k_n),\ttt_1,\ttt_2)$ is a binary
operator that computes a \tg on the vertices $\tv = \tv_1 \cup^T
\tv_2$, with edges and edge attributes computed by a conjunctive query
over the constituent relations of $\ttt_1$ and $\ttt_2$.  
$L(v_1,v_2,a,p) = \resolve{f_1(k_1),\ldots,f_n(k_n)}{q(\ttt_1,
  \ttt_2)}$, and then set $\te' = \coal{\pi^T_{v1,v2} L}$, and $\tae'
= \sigma_{a~is~not~null} L$.  We then compute $\tav =
\resolve{f_1(k_1),\ldots,f_n(k_n)}{\tv_1 \cup^T \tv_2}$.

Edge creation has several important applications.  It can be used to
compute friend-of-friend edges (passing in the same \tg as both
arguments).  Since $q$ can include predicates over the timestamps,
$\insql{edge}^T$ can compute journeys.  A journey is a path in the
evolving graph with non-decreasing time
edges~\cite{Casteigts2011,Ferreira2004}.  By adding a temporal
condition to $q$, we can obtain journeys similar to time-concurrent
paths.

In graph theory, a graph join of two undirected unlabeled disjoint
graphs is defined as the union of the two graphs and additional edges
connecting every vertex in graph one with each vertex in graph two.
We can obtain a graph join by computing $\te' = \tv_1 \times^T \tv_2$.

SocialScope~\cite{Amer-Yahia2009} defines (non-temporal) graph
composition: compose the edge of the two operands and return an
edge-induced subgraph.  TGA can express this operator by a combination
of edge creation and vertex subgraph.

\subsection{Extension: user-defined analytics}
\label{sec:analytics}

For many types of analysis, it is necessary to compute some property,
such as PageRank of each vertex $v$, or the length of the shortest
path from a given designated vertex $u$ to each $v$, for each time
point.  This information can then be used to study how the graph
evolves over time.  \ql supports this type of analysis through {\em
  temporal user-defined analytics}, which conceptually execute an
aggregation-map operation sequence repeatedly over a set number of
iterations or until fix-point and save the result in property $pname$.
$\ttt' = \insql{pregel}^T (dir,cond,f_m,f_a,pname,iter,\ttt)$, where
all arguments are like in aggregation, with an additional $iter$
argument specifying the number of iterations to perform.

\section{Expressive power}
\label{sec:formal}

In this section we study expressiveness of the \tg model, which
consists of the \tg data structure (Definition~\ref{def:tg}) and of
\tga, an algebra for querying the data structure
(Section~\ref{sec:algebra}). We stress that ours is a valid-time data
model that does not provide transaction-time and bi-temporal support.

{\bf Important note:} We restrict our attention to a subset of \tga
operations, excluding window-based node creation
(Section~\ref{sec:algebra:ncreate}) from our analysis.  Window-based
node creation requires the split $\wsplit{s}{w}{R}$ primitive, which
cannot be naturally expressed in \tra.  We defer an investigation of
expressiveness of \tga with window-based node creation to future work.

We start by proposing two natural notions of completeness for a
temporal graph query language.

\begin{definition}
  Let $L^t$ be a temporal relational language and $\tve$ --- a
  relational representation of a temporal graph.  An \edgeq $q^t_e$ in
  $L$ takes a graph $\tve (\tv, \te, \tav, \tae)$ as input, and
  outputs another graph $\tve'$ on the vertices of $\tve$ such that
  the edges of $\tve'$ are defined by $q^t_e$.  A language is
  $L^t$-\edgec if it can express each $q^t_e$ in $L^t$.
  \label{def:edgecomplete}
\end{definition}

Note that the query $q^t_e$ is not restricted to act on \te alone, and
may refer to the other constituent relations \tve.

\begin{definition}
  Let $L^t$ be a temporal relational language, and let $\tve$ be a
  relational representation of a temporal graph.  A \vertexq $q^t_v$
  in $L^t$ takes a graph $\tve (\tv, \te, \tav, \tae)$ as input, and
  outputs another graph $\tve'$ such that the vertices of $\tve'$ are
  defined by $q^t_v$. A language is $L^t$-\vertexc if it can express
  each $q^t_v$ in $L^t$.
\label{def:vertexcomplete}
\end{definition}

We now refer to definitions~\ref{def:edgecomplete}
and~\ref{def:vertexcomplete} and show that \tga is \edgec and
\vertexc, with respect to the valid-time fragment of temporal
relational algebra (\tra).  \tra is an algebra that corresponds to
temporal relational calculus~\cite{DBLP:reference/db/ChomickiT09b}, a
first-order logic that extends relational calculus, supporting
variables and quantifiers over both the data domain and time domain.

\begin{theorem}
\tga is TRA-\edgec.
\label{th:edgecomplete}
\end{theorem}

\begin{proof}
  The result of every conjunctive edge-query over the vertices of \ttt
  can be expressed by $\sigma_{c} (\tv \times^T \tv)$.  Queries of
  this kind can be expressed by the edge creation operator of \tga
  (Section~\ref{sec:algebra:ecreate}), invoked as:\\
  $\insql{edge}^T(q=\sigma_{c} (\ttt_1.\tv \times^T \ttt_2.\tv),\ttt_1
  = \ttt,\ttt_2=\ttt)$\end{proof}

\begin{theorem}
\tga is TRA-\vertexc.
\label{th:vertexcomplete}
\end{theorem}

\begin{proof}
  Every \tra vertex-query can be expressed in \tga
  by a sequence of vertex-subgraph $q^T_v(\tve)$
  (Section~\ref{sec:algebra:unary}) and attribute-based node creation
  $\insql{node}^T_a$ (Section~\ref{sec:algebra:ncreate}).
  Attribute-based node creation supports Skolem functions, and is
  necessary to handle queries that introduce vertex identifiers.  
\end{proof}

For a point-based model, it is customary to interrogate two
properties --- snapshot reducibility (S-reducibility) and extended
snapshot reducibility (extended S-reducibility).
S-reducibility states that for every query $q$ in $L$, there must
exist a syntactically similar query $q^t$ in $L^t$ that generalizes
$q$.  Specifically the following relationship should hold when $q^t$
is evaluated over a temporal database $D^t$ (recall that $\tau$ is the
temporal slice operator): $q(\tau_c(D^t)) = \tau_c(q^t(D^t))$, for all
time points $c$.  Extended S-reducibility requires that $L^t$ provide
an ability to make explicit references to timestamps alongside
non-temporal predicates.

\tga is s-reducible and extended s-reducible because, as we showed in
Section~\ref{sec:algebra}, every \tra operation can be rewritten into
\tra, which is s-reducible and extended s-reducible w.r.t. relational
algebra.

\section{System}
\label{sec:sys}

We developed a prototype system \ql which supports \tga operations on
top of Apache Spark/GraphX~\cite{DBLP:conf/osdi/GonzalezXDCFS14}.  The
data is distributed in partitions across the cluster workers, read in
from HDFS, and can be viewed both as a graph and as a pair of RDDs.
All \tg operations are available through the public API of the \ql
library, and may be used in an Apache Spark application.

\subsection{Reducing temporal operators}

Apache Spark is not a temporal DBMS but rather an open-source
in-memory distributed framework that combines graph parallel and data
parallel abstractions.  Following the approach of Dignos et
al.~\cite{Dignos2012} we reduce our temporal operators into a sequence
of nontemporal relational operators or their equivalents for Spark
RDDs, maintaining point semantics.  This allows our algebra to be
implemented in any nontemporal relational database.  In total, we need
the four temporal primitives we introduced in
Section~\ref{sec:algebra:integrity} (coalesce, resolve, constrain, and
split), as well as the primitives described in~\cite{Dignos2012}:
extend and normalize.  Because our model uses point semantics and does
not require change preservation, we do not need the align primitive
of~\cite{Dignos2012} and can use the normalize primitive in its place.

The {\em coalesce} primitive merges adjacent and overlapping time
periods for value-equivalent tuples.  This operation, which is similar
to duplicate elimination in conventional databases, has been
extensively studied in the
literature~\cite{DBLP:conf/vldb/BohlenSS96,DBLP:journals/sigmod/Zimanyi06}.
Several implementations are possible for the coalesce operation over
temporal SQL relations.  Because Spark is an in-memory processing
system, we use the partitioning method, where the relation is grouped
by key, and tuples are sorted and folded within each group to produce
time periods of maximum length.  Eager coalescing, however, is not
desirable since it is expensive and some operations may produce
correct results (up to coalescing) {\em even when computing over
  uncoalesced inputs}.  Any operation that is time-variant
requires input to be coalesced.  We base the eager coalescing
rules on coalescing rules in \tra~\cite{DBLP:conf/vldb/BohlenSS96}.

The {\em resolve} primitive is implemented using a group by key
operation in Spark and convenience methods on the property set class.
The property set class supports adding all properties from another set
such that they are combined by key, and applying aggregation functions
one at a time for each property name.

The {\em constrain} primitive constrains one relation with respect to
another, such as removing edges from the result that do not have
associated nodes, or trimming the edge validity period to be within
the validity periods of associated nodes.  It is introduced here
because Spark does not have a built-in way to express foreign key
constraints.  We do this by executing a join of the two relations ---
either a broadcast join or a hash join --- and then adjusting time
periods as necessary.  This is an expensive operation and is only
performed when necessary as determined by the soundness analysis,
e.g., when vertex-subgraph has a non-trivial predicate over \tve, and
when node creation has a more restrictive vertex quantifier $q_v$ than
edge quantifier $q_e$.

The {\em split} primitive maps each tuple in relation $R$ into one or
more tuples based on a temporal window expression such as
\insql{w=3~months}.  A purely relational implementation of this
primitive is possible with the use of a special Chron relation that
stores all possible time points of the temporal universe and supports
computation without materialization.  Another approach is to introduce
fold and unfold functions that can split each interval into all its
constituent time points.  Both of these approaches have strong
efficiency concerns, see~\cite{DBLP:conf/time/BohlenGJ06} for an
in-depth discussion.  In Spark we are not limited to relational
operators only and can use functional programming constructs.  Split
can be efficiently implemented with a \insql{flatMap}, which emits
multiple tuples as necessary by applying a lambda function and
flattening the result.  We use this method in our implementation.

The {\em extend} primitive extends a relation with an additional
attribute that represents the tuple's timestamp, see~\cite{Dignos2012}
for a definition.  Extend allows explicit references to timestamps in
operations, and is needed for extended snapshot reducibility.  We
implement extend by defining an Interval class and including it as a
field in every RDD.

The {\em normalize} primitive produces a set of tuples for each tuple
in {\bf r} by splitting its timestamp into non-overlapping periods
with respect to another relation {\bf s} and attributes {\bf B}.
See~\cite{Dignos2012} for the formal definition.  Intuitively,
normalize creates tuples in corresponding groups such that their
timestamps are also equivalent.  This primitive is necessary for node
creation, set operators like union, and joins.  Normalize primitive
relies on an efficient implementation of the tuple splitter.  We split
each tuple based on the change periods over the whole graph, avoiding
costly joins but potentially splitting some tuples unnecessarily.

\subsection{Physical Representations}
\label{sec:sys:datastructs}

It is convenient to use intervals to compactly represent consecutive
value-equivalent snapshots of \tve --- timeslices in which no change
occurred in graph topology, or in vertex and edge attributes.  We use
the term {\em representative graph} to refer to such snapshots, since
they represent an interval. 

We considered four in-memory \tg representations that differ both in
compactness and in the kind of locality they prioritize. With {\em
  structural locality}, neighboring vertices (resp. edges) of the same
representative graph are laid out together, while with {\em temporal
  locality}, consecutive states of the same vertex (resp. edge) are
laid out together~\cite{Miao2015}.  We now
describe each representation.

We can convert from one representation to any other at a small cost
(as supported by our experimental results), so it is useful to think
of them as access methods in the context of individual operations.

{\bf VertexEdge (VE)} is a direct implementation of the \tve model,
and is the most compact: one RDD contains all vertices and another all
edges.  Consistently with the GraphX API, all vertex properties are
stored together as a single nested attribute, as are all edge
properties.  We currently do not store the \tv and \te relations
separately but rather together with \tav and \tae, respectively.
While VE does not necessitate a particular order of tuples on disk, we
opt for a physical layout in which all tuples corresponding to the
same vertex (resp. edge) are laid out consecutively, and so VE
preserves temporal locality.

VE supports all \tga operations except analytics, because an analytic
is defined on a representative graph, which VE does not materialize.
As we will show in Section~\ref{sec:exp}, this physical representation
is the most efficient for many operations.  In the current prototype
we limit the expressiveness of some of the operations, such as only
supporting vertex- and edge-subgraph queries over the \tav and \tae
relations.

{\bf RepresentativeGraphs (RG)} is a collection (parallel sequence) of
GraphX graphs, one for each representative graph of \ttt, where
vertices and edges store the attribute values for the specific time
interval, thus using structural locality.  This representation
supports all operations of \tga which can be expressed over snapshots,
i.e. any operation which does not explicitly refer to time.  GraphX
provides Pregel API which is used to support all the analytics.
While the \rg representation is simple, it is not compact, considering
that in many real-world evolving graphs there is a 80\% or larger
similarity between consecutive snapshots~\cite{Miao2015}.  In a
distributed architecture, however, this data structure provides some
benefits as operations on it can be easily parallelized by assigning
different representative graphs to different workers.  We include this
representation mainly as a naive implementation to compare performance
against.

\rg is the most immediate way to implement evolving graphs using
GraphX. Without \ql a user wishing to analyze evolving graphs might
implement and use the \rg approach.  However, as we will show in
Section~\ref{sec:exp}, this would lead to poor performance for most
operations.

\begin{figure}[t!]
\centering
\includegraphics[width=3in]{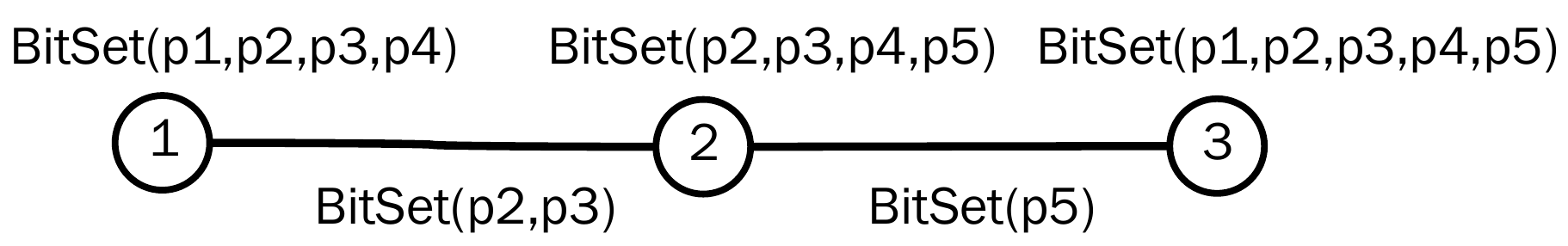}
\caption{\og representation of \insql{T1}.}
\label{fig:ogc}
\end{figure}

{\bf OneGraph (OG)} is the most topologically compact representation,
which stores all vertices from \tav {\em and} edges from \tae once, in
a single aggregated data structure.  OG emphasizes temporal locality,
while also preserving structural locality, but leads to a much denser
graph than RG.  This, in turn, makes parallelizing computation
challenging.

An OG is implemented as a single GraphX graph where the vertex and
edge attributes are bitsets that encode the presence of a vertex or
edge in each time period associated with some representative graph of
a \tg.  To construct an \og from \tve, vertices and edges of \tv and
\te relations each are grouped by key and mapped to bits corresponding
to periods of change over the graph.  Because \og stores information
only about graph topology, far fewer periods must be represented and
computed for \og than for \rg.  The actual reduction depends on the
rate and nature of graph evolution.  Information about time validity
is stored together with each vertex and edge.  Figure~\ref{fig:ogc}
shows the OG for \insql{T1} from Figure~\ref{fig:tg_ve}.

Analytics are supported using a batching method over the Pregel API.
Similar to ImmortalGraph~\cite{Miao2015}, the analytics are computed
over all the representative graphs together.  Vertices exchange
messages marked with the applicable intervals and a single message may
contain several interval values as necessary.

As we will see experimentally in Section~\ref{sec:exp}, \og is often
the best-performing data structure for node creation, and also has
competitive performance for analytics.  Because of this focus, \og
supports operations only on topology: analytics, node creation. and
set operators for graphs with no vertex or edge attributes.  All other
operations are supported through inheritance from an abstract parent,
and are carried out on the VE data structure.  Thus \og and \hg,
below, can be thought of as indexes on VE.

In our preliminary experiments we observed that \og exhibited worse
than expected performance, especially for large graphs with long
lifetimes.  The reason this is so is because good graph partitioning
becomes difficult as topology changes over time.  Communication cost
is the main contributor to analytics performance over distributed
graphs, so poor partitioning leads to increased communication costs.
When the whole graph can fit into memory of a single worker,
communication cost goes away and the batching method used by \og
becomes the most efficient, as has been previously shown
in~\cite{Miao2015}.  To provide better performance on analytics, we
introduce {\bf HybridGraph (HG)}.  \hg trades compactness of \og for
better structural locality of \rg, by aggregating together several
consecutive representative graphs, computing a single \og for each
graph group, and storing these as a parallel sequence.  In our current
implementation each \og in the sequence corresponds to the same number
of temporally adjacent graphs.
This is the simplest grouping method, and we observed that placing the
same number of graphs into each group often results in unbalanced
group sizes.  This is because evolving graphs commonly exhibit strong
temporal skew, with later graphs being significantly larger than earlier
ones.  We are currently working on more sophisticated grouping
approaches that would lead to better balance, and ultimately to better
performance.  However as we will see experimentally in
Section~\ref{sec:exp}, the current \hg implementation already improves
performance compared to \og, in some cases significantly.

Like \og, \hg focuses on topology-based analysis, and so does not
represent vertex and edge attributes. \hg implements analytics, node
creation, and set operators, and supports all other operations through
inheritance from VE.  Analytics are implemented similar to \og, with
batching within each graph group.

Since RG, OG and HG are implemented over GraphX graphs, the
referential integrity is maintained by the framework and the constrain
primitive is not required.  All primitives are used with the VE
representation.

\subsection{Additional Implementation Details}
\label{sec:sys:maint}

{\bf Partitioning.}  Graph partitioning has a tremendous impact on
performance.  A good partitioning strategy needs to be balanced,
assigning an approximately equal number of units to each partition,
and limit the number of cuts across partitions, reducing
cross-partition communication.  In previous experiments we compared
performance with no repartitioning after load vs.  with
repartitioning, using the GraphX E2D edge partitioning strategy.  In
E2D, a sparse edge adjacency matrix is partitioned in two dimensions,
guaranteeing a $2 \sqrt{n}$ bound on vertex replication, where $n$ is
the number of partitions. E2D has been shown to provide good
performance for Pregel-style
analytics~\cite{DBLP:conf/osdi/GonzalezXDCFS14,MoffittTempWeb16}.  The
user can repartition the representations at will, consistent with the
Spark approach.

{\bf Graph loading.}  We use the Apache Parquet format for on-disk
storage, with one archive for vertices and another for edges,
temporally coalesced.  This format corresponds to the VE physical
representation~\ref{sec:sys:datastructs}.  In cases where there is no
more than 1 attribute per vertex and edge, this representation is also
the most compact.  

For ease of use, we provide a GraphLoader utility that can initialize
any of the four physical representations from Apache Parquet files on
HDFS or on local disk.  A \ql user can also
implement custom graph loading methods to load vertices and edges, and
then use the \insql{fromRDDs} to initialize any of the four physical
representations.

{\bf Integration with SQL.}  The \ql API exposes vertex/edge RDDs to
the user and provides convenience methods to convert them to Spark
Datasets.  Arbitrary SparkSQL queries can then be executed over these
relations.

\section{Experimental Evaluation}
\label{sec:exp}

\begin{figure*}[t]
\vspace{-0.2in}
\centering
\begin{minipage}[b]{2.1in}
\centering
\includegraphics[width=2.1in]{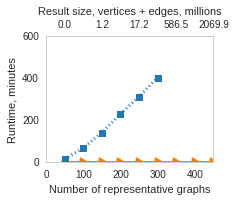}
\vspace{-0.2in}
\caption{Slice on nGrams.}
\label{fig:slicengrams}
\vspace{-0.1in}
\end{minipage}
\begin{minipage}[b]{2.1in}
\centering
\includegraphics[width=2.1in]{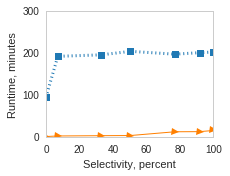}
\vspace{-0.2in}
\caption{Subgraph on nGrams.}
\label{fig:subgraphngrams}
\vspace{-0.1in}
\end{minipage}
\begin{minipage}[b]{2.1in}
\includegraphics[width=2.4in]{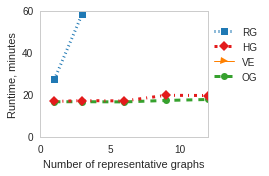}
\vspace{-0.2in}
\caption{Aggregate on Twitter.}
\label{fig:degtwitter}
\vspace{-0.1in}
\end{minipage}
\end{figure*}

\begin{figure*}[t!]
\centering
\begin{subfigure}{0.3\textwidth}
\includegraphics[width=2.2in]{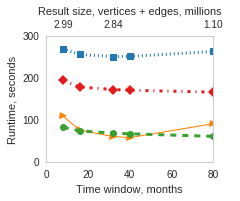}
\caption{$q_v=\insql{always}$, $q_e=\insql{always}$, wiki-talk}
\label{fig:agg1}
\end{subfigure}
\begin{subfigure}{0.3\textwidth}
\includegraphics[width=2.2in]{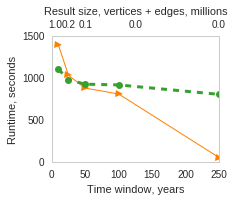}
\caption{$q_v=\insql{always}$, $q_e=\insql{exists}$, nGrams}
\label{fig:agg2}
\end{subfigure}
\begin{subfigure}{0.35\textwidth}
\includegraphics[width=2.5in]{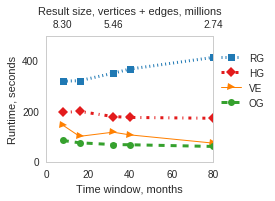}
\caption{$q_v=\insql{always}$, $q_e=\insql{exists}$, wiki-talk}
\label{fig:agg4}
\end{subfigure}
\caption[]{Node creation with temporal windows.}
\vspace{-0.1in}
\label{fig:agg}
\end{figure*}

\subsection{Setup}
\label{sec:exp:setup}

All experiments were conducted on a 16-slave in-house Open Stack
cloud, using Linux Ubuntu 14.04 and Spark v2.0.  Each node has 4 cores
and 16 GB of RAM.  Spark Standalone cluster manager and Hadoop 2.6
were used.
Because Spark is a lazy evaluation system, a materialize operation was
appended to the end of each query, which consisted of the count of
nodes and edges.  Each experiment was conducted 3 times with a cold
start -- the running time includes the setup time of submitting the
job to the cluster manager, uploading the jar to the cluster, reading
the data from disk, building the chosen representation, and running
a single query.  We report the average running time, which is
representative because we took great care to control variability:
standard deviation for each measure is at or below 5\% of the mean
except in cases of very small running times.  No computation results
were shared between subsequent runs.

{\bf Data.}  We evaluate performance of our framework on three real
open-source datasets, summarized in Table~\ref{tab:datasets}.
wiki-talk (\url{http://dx.doi.org/10.5281/zenodo.49561}) contains over
10 million messaging events among 3 million wiki-en users\eat{2002
  through 2015}, aggregated at 1-month resolution.\\nGrams
(\url{http://storage.googleapis.com/books/ngrams/books/datasetsv2.html})
contains word co-occurrence pairs\eat{ from 1520 through 2008}, with
30 million word nodes and over 2.5 billion undirected edges.  The
Twitter social graph~\cite{Gabielkov:2014:SSN:2591971.2591985}
contains over 23 billion directed follower relationships between 0.5
billion twitter users\eat{ collected in 2012}, sampled at 1-month
resolution.\eat{ based on account creation information from April
  2006.} \eat{DELIS contains monthly snapshots of a portion of the Web
  graph focusing on the .uk domains from 05/2006 through
  05/2007~\cite{BSVLTAG}. }The datasets differ in size, in the number
and type of attributes and in evolution rates, calculated as the
average graph edit similarity~\cite{Ren2011}. \eat{the evolutionary
  properties: co-authorship network nodes and edges have limited
  lifespan, while the nGrams network grows over time, with nodes and
  edges persisting for long duration.  All figures in the body of this
  section are on the larger nGrams dataset.  Refer to the Appendix for
  the DBLP figures, which show similar trends as nGrams.}

\subsection{Individual operators}
\label{sec:exp:ops}

{\bf Slice} performance was evaluated by varying the slice time window
and materializing the \tg, and is presented in
Figures~\ref{fig:slicengrams} for nGrams and~\ref{fig:slicewiki} for
wiki-talk (in Appendix~\ref{sec:app2}).  Similar trends were observed
for twitter.  Slice is expected to be more efficient when executed
over VE when data is coalesced on disk than over \sg, and we observe
this in our experiments.  This is because multiple passes over the
data are required for \sg to compute each representative graph,
leading to linear growth in running times for file formats and systems
without filter pushdown, as is the case here.  Slice over VE simply
executes temporal selection and has constant running times (29 sec for
wiki-talk, about 1.5 min for nGrams).\eat{Recall that in VE
  \insql{slice} performs temporal selection, and method when data on
  disk is coalesced.  \sg, in contrast, does multiple passes of select
  over the same data to compute each RG.  Thus, as expected, VE
  behavior is directly dependent on the size of input data regardless
  of the slice size for file formats and systems without filter
  pushdown, as is the case here, or if the data does not have temporal
  locality (Figure~\ref{fig:slicewiki}, Figure~\ref{fig:slicengrams}
  -- about 1.5 minutes for VE).  \sg behavior is linear in the slice
  size. } This experiment essentially measures the cost of
materializing \sg from its on-disk representation. \eat{ We observed the
same linear trend in our preliminary work, when the data was stored as
individual snapshots on disk, although less redundant work is needed
in that case.}

{\bf Vertex subgraph} performance was evaluated by specifying a
condition on the $length(a.attr)<t$ of the vertex attribute, with
different values of $t$ leading to different selectivity.  This
experiment was executed for wiki-talk (with $username$ as the
property) and for nGrams (with $word$ as the property).  Twitter has
no vertex attributes and was not used in this experiment.
Figure~\ref{fig:subgraphngrams} shows performance for \sg and VE on
nGrams (wiki-talk results in Appendix~\ref{sec:app2}).  Performance on
\sg is a function of the number of intervals and is insensitive to the
selectivity.  The behavior on VE is dominated by FK enforcement:
with high selectivity (few vertices) broadcast join affords
performance linear in the number of edges, whereas for a large number
of vertices broadcast join is infeasible and a hash-join is used
instead, which is substantially slower.  VE provides an order of
magnitude better performance than \sg: up to 3 min with hash-join and
up to 15 min with broadcast join for VE, in contrast to between 95 and
200 min for \sg.

{\bf Map}\eat{ We evaluate \insql{map} performance by varying the data
  size through the \insql{slice} operation.} exhibits a similar trend
as \insql{slice}: constant running time for VE and a linear increase
in running time with increasing number of representative graphs for
\sg (Figure~\ref{fig:project} for wiki-talk in
Appendix~\ref{sec:app2}, similar for other datasets).  Performance of
\insql{map} is slightly worse than that of \insql{slice} because
\insql{map} must coalesce its output as the last step, while
\insql{slice} does not.

\begin{table}
\caption{Experimental datasets.}
\vspace{-0.1in}
\small
\begin{tabular}{l | c | c | c | c }
\hline
\multicolumn{1}{l|}{\bfseries Dataset} & \multicolumn{1}{c|}{\bfseries |V|} & \multicolumn{1}{c|}{\bfseries |E|} & \multicolumn{1}{c|}{\bfseries Time Span} & \multicolumn{1}{c}{\bfseries Evol. Rate} \\ \hline
wiki-talk-en & 2.9M & 10.7M & 2002--2015 & 14.4 \\ \hline
nGrams & 29.3M & 2.5B & 1520--2008 & 16.67 \\ \hline
twitter & 505.4M & 23B & 2006--2012 & 88 \\ \hline
\end{tabular}
\vspace{-0.1cm}
\label{tab:datasets}
\end{table}

\begin{figure*}[t]
\begin{minipage}[b]{2.1in}
\centering
\includegraphics[width=2.1in]{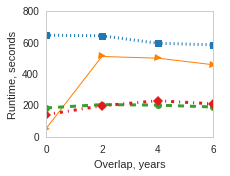}
\vspace{-0.2in}
\caption{Union on wiki-talk.}
\label{fig:union1}
\vspace{-0.1in}
\end{minipage}
\begin{minipage}[b]{2.4in}
\centering
\includegraphics[width=2.4in]{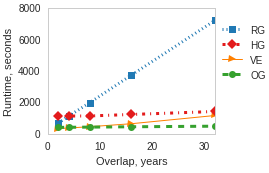}
\vspace{-0.2in}
\caption{Intersection on nGrams.}
\label{fig:intersectngrams}
\vspace{-0.1in}
\end{minipage}
\begin{minipage}[b]{2.2in}
\centering
\includegraphics[width=2.4in]{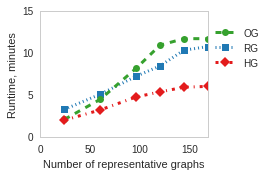}
\vspace{-0.2in}
\caption{Components on wiki-talk.}
\label{fig:ccwiki}
\vspace{-0.1in}
\end{minipage}
\end{figure*}

{\bf Aggregation} performance was evaluated on the graph-based
representations with a computation of vertex degrees, varying the size
of the temporal window obtained with slice.  The results in
Figure~\ref{fig:degtwitter} indicate that materialization of each
representative graph required for \sg makes it not a viable candidate
for this operation, especially over large datasets.  Both \og and \hg
exhibit linear increase in performance as the slice size is increased,
with a small slope.  Similar performance was observed in the other
datasets.

{\bf Node creation} performance was evaluated on all representations,
since all have different implementations of this operator.  We
executed topology-only creation (no attributes), varying the size
of the temporal window.  We observe that performance depends heavily
on the quantification, and on the data evolution rate.  \og is an
aggregated data structure with good temporal locality and thus in most
cases provides good performance and is insensitive to the temporal
window size (Figure~\ref{fig:agg1}).  However, in datasets with a
large number of representative graphs (such as nGrams), \og is slow on
large windows, an order of magnitude worse than VE in the worst case
(Figure~\ref{fig:agg2}).  VE outperforms \og when vertex and edge
quantification levels match (Figure~\ref{fig:agg1}), but is worse than
\og when vertex quantification is stricter than edge quantification
and FK must be enforced (Figure~\ref{fig:agg4}).  \og also outperforms
VE when both evolution rate is low and aggregation window is small
(Figure~\ref{fig:agg1}, wiki-talk).\eat{ \sg and \hg do not provide
  the best performance on any of our datasets.}

{\bf Union, intersection, and difference} by structure were evaluated
by loading two time slices of the same dataset with varying temporal
overlap.  Performance depends on the size of the overlap (in the
number of representative graphs) and on the evolution rate.  VE has
best performance when overlap is small (Figure~\ref{fig:union1})\eat{
  and when the evolution rate is high (Figure~\ref{fig:union2}),
  regardless of the size of the overlap}.  \og always has good
performance, constant w.r.t. overlap size.  This is expected, since
\og \insql{union} and \insql{intersection} are implemented as joins
(outer or inner) on the vertices and edges of the two operands.  VE,
on the other hand, splits the coalesced vertices/edges of each of the
two operands into intervals first, takes a union, and then reduces by
key.  When evolution rate is low and duration of an entity is high,
such as in wiki-talk for vertices, the split produces a lot of tuples
to then reduce, and performance suffers (Figure~\ref{fig:union1}). \sg
only has good performance on \insql{intersection} when few
representative graphs overlap, and never on \insql{union}
(Figure~\ref{fig:intersectngrams}). \hg performance is worse than \og,
by a constant amount in \insql{union}, and diverges in
\insql{intersection}.  \eat{\insql{difference} performance is similar
  to \insql{intersection}.}

{\bf Analytics.}  We implemented PageRank (PR) and Connected
Components (CC) analytics for the three graph-based representations
using the Pregel GraphX API.  PR was executed for 10 iterations or
until convergence, whichever came first. CC was executed until
convergence with no limit on the number of iterations.  Performance of
Pregel-based algorithms depends heavily on the partitioning strategy,
with best results achieved where cross-partition communication is
small~\cite{MoffittTempWeb16}.  For this reason, we evaluated only
with the E2D strategy.
Performance was evaluated on time slices of varying size.  Recollect
that analytics are essentially multiple rounds of aggregate
operations, so the performance we observe is an amplified version of
aggregate performance.  For a very small number of graphs (1-2), \sg
provides good performance, but slows down linearly as the number of
graphs increases.  \hg provides the best performance on analytics
under most conditions, with a linear increase but a significantly
slower rate of growth.  The tradeoff between \og and \hg depends on
graph evolution characteristics.  If the graph is a growth-only
evolution (such as in Twitter), \og is not denser than \hg and
computes everything in a single batch, which leads to the fastest
performance, as can be seen in Figure~\ref{fig:pranktwitter}.  If the
edge evolution represents more transient connections, then \hg is less
dense and scales better (Figure~\ref{fig:ccwiki}).  Note that \og and
\hg performance could be further improved by computing them over
coalesced structure-only V and E, and ignoring attributes.

{\bf In summary,} no one data structure is most efficient across all
operations.  This opens the door to query optimization based on the
characteristics of the data such as graph evolution rate and on the
type of operation being performed.  

\subsection{Switching between representations}
\label{sec:exp:scale}

\begin{figure*}[t]
\centering
\begin{minipage}{2.2in}
\centering
\includegraphics[width=2.1in]{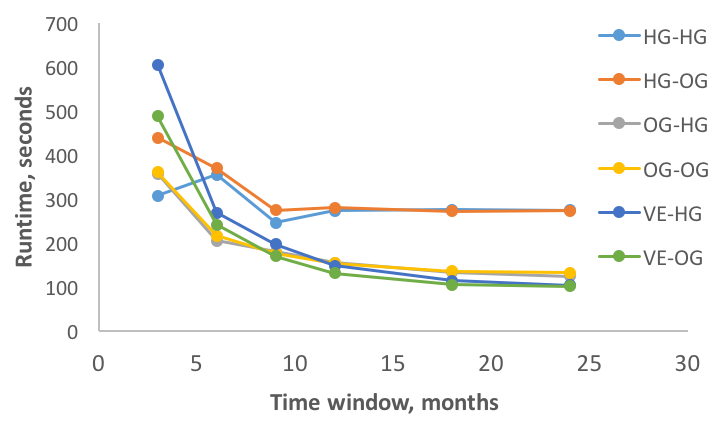}
\caption{$\insql{node}^T_a$ followed by components.}
\label{fig:ncrtocc}
\vspace{-0.1in}
\end{minipage}
\begin{minipage}{2.2in}
\includegraphics[width=2.2in]{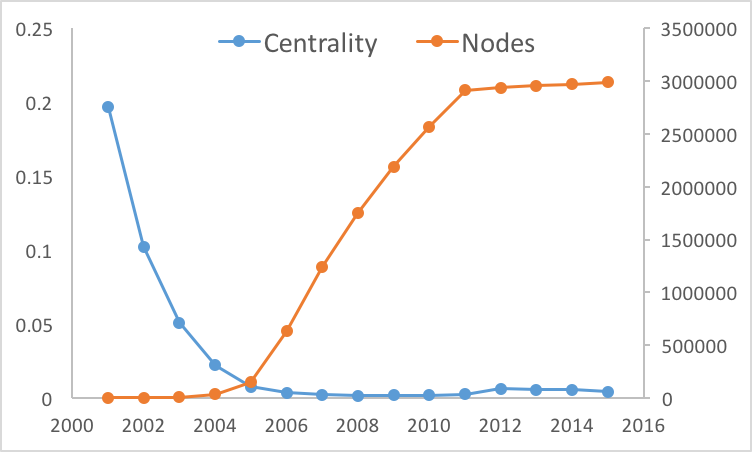}
\caption{In-degree centrality with 1 year resolution.}
\label{fig:central}
\end{minipage}
\begin{minipage}{2.2in}
\includegraphics[width=2.2in]{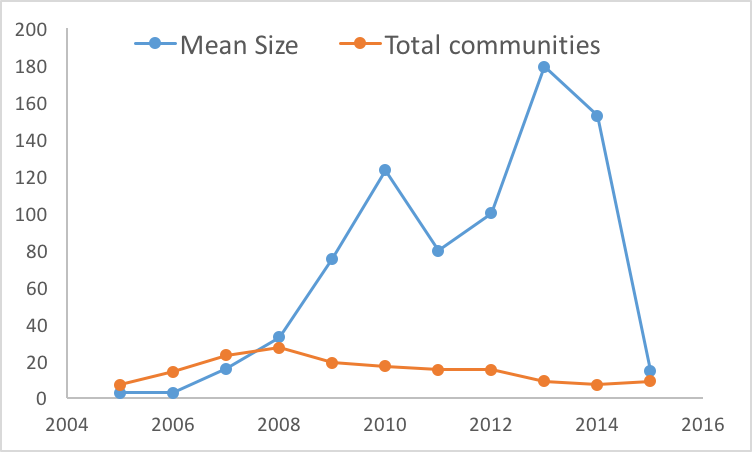}
\caption{Communities with 1 year resolution.}
\label{fig:commun}
\end{minipage}
\end{figure*}

To treat the four representations as access methods, we need to be
able to switch between them.  The data structures can be created from
outputs of any of the four, at a cost.  To investigate the feasibility
of switching between representations, we executed two-operator queries
and either kept the representation constant or changed it between the
operators.  The query is based on the first two steps of example three
in our motivating use cases: node creation over temporal windows in
wiki-talk, followed by the connected components analytic.
Figure~\ref{fig:ncrtocc} shows the result of varying the size of the
temporal window.  Recall that \og is the best performing
representation for node creation at small windows and \hg for
components over this dataset.  The benefits of \hg are substantially
consumed by switching: the performance of \og-\og and and \og-\hg are
similar.  If the cost of switching was negligible, then \og-\hg should
have exhibited notably better performance than all other combinations.
However, the \og-\hg still performs best over-all, indicating that
switching is feasible.

\subsection{Use cases}
\label{sec:exp:cases}

To see how our algebra handles the use cases from
Section~\ref{sec:cases}, we implemented each one over the wiki-talk
dataset.  Each example requires a sequence of operators.  For each
operator we used the best performing data structures based on the
comparison experiments described above.

Example 1 answers the question of whether there are high influence
nodes and whether that behavior is persistent in time.  The code to
compute the answer is 4 lines of a Scala program and the query took 76
seconds to execute.  The results show that from 25 nodes with mean
degree of 40 and above that have persisted for at least 6 months, 6
have coefficient of variation below 50, which is quite low, and only 5
have it above 100.  This indicates that there are in fact high
in-degree nodes and that they continue to be influential over long
periods of time, despite the loose connectivity of the overall
network.

Example 2 examines how the graph centrality changes over time.  The
program is 6 lines of Scala code iterating with temporal windows of 1,
2, 3, 6, and 12 months, and the analysis took 25 minutes.  Results
show that regardless of the temporal resolution, the in-degree
centrality is extremely low, about 0.04.  Figure~\ref{fig:central}
provides an explanation -- as the size of the graph increases, its
centrality decreases.  Given that the number of edges in this graph is
only about 4 times the number of nodes, the graph is too sparse and
disjointed to have any centrality.

Finally, example 3 examines whether communities can be detected in the
wiki-talk network at different temporal resolution.  The program,
similar to the one above, is 6 lines of Scala code with varied
temporal windows.  The total runtime is 58 minutes.  Communities,
defined as connected components, can be detected in all temporal
resolutions.  As a reminder, the edge quantification in this query is
\insql{always}, so only edges that persist over each window are
retained.  The presence of communities even with large temporal
resolution indicates that communities form and persist over time.
Figure~\ref{fig:commun} shows the mean size of all communities by time
and their total number.  The peaks of the mean size, visible in all
temporal windows, may indicate that communities form and then reform
in a different configuration, perhaps for a different purpose.  The
results of this analysis can serve as a starting point to investigate
the large communities and what caused the size shifts.

{\bf In summary,} complex analyses can be expressed as queries in \ql
and lead to interesting insights about the evolution of the underlying
phenomena.

\section{Related Work}
\label{sec:related}

{\bf Evolving graph models.}  Much recent work
represents evolving graphs as sequences of snapshots in a discrete
time domain, and focuses on snapshot retrieval and
analytics~\cite{Khurana2013,Miao2015,Ren2011}.  Our logical model is
semantically equivalent to a sequence of snapshots because in point
semantics snapshots can be obtained with a simple slice over all time
points.  We choose to represent \tgs as a collection of vertices and
edges because the range of operations we support is naturally
expressible over them but not over a sequence of snapshots.  For
example, subgraph with a temporal predicate is impossible to express
over snapshot sequences as each snapshot is nontemporal and
independent of the others.

{\bf Querying and analytics.} There has been much recent work on
analytics for evolving graphs,
see~\cite{DBLP:journals/csur/AggarwalS14} for a survey. This line of
work is synergistic with ours, since our aim is to provide systematic
support for scalable querying and analysis of evolving graphs.

Several researchers have proposed individual queries, or classes of
queries, for evolving graphs, but without a unifying syntax or general
framework.  The proposed operators can be divided into those that
return temporal or nontemporal result.  Temporal operators include
retrieval of version data for a particular node and
edge~\cite{George2006}, of journeys~\cite{George2009,Casteigts2011},
subgraph by time or attributes~\cite{Huo2014,Khurana2016}, snapshot
analytics~\cite{Miao2015,Labouseur2015,Khurana2016}, and computation
of time-varying versions of whole-graph analytics like maximal
time-connected component~\cite{Ferreira2004} and dynamic graph
centrality~\cite{Lerman2010}.  Non-temporal operators include snapshot
retrieval~\cite{Khurana2013} and retrieval at time
point~\cite{George2009,Khurana2016}.

Our contribution is to propose one unifying compositional \tga that
covers the range of the operations in a complete way, with clear
semantics, which many previous works lack.

{\bf Implementations.}  Three systems in the literature focus on
systematic support of evolving graphs, all of them non-compositional.
Miao et al.~\cite{Miao2015} developed ImmortalGraph (formerly
Chronos), a proprietary in-memory execution engine for temporal graph
analytics.  The ImmortalGraph system has no formal model, but
informally an evolving graph is defined as a series of activities on
the graph, such as node additions and deletions.  This is a streaming
or delta approach, which is popular in temporal databases because it
is unambiguous and compact.  ImmortalGraph does not provide a query
language, focusing primarily on efficient physical data layout.  Many
insights about temporal vs. structural locality by~\cite{Miao2015}
hold in our setting.  The batching method for snapshot analytics used
by \og is similar to the one proposed in ImmortalGraph.  However,
ImmortalGraph was developed with the focus on centralized rather than
distributed computation and~\cite{Miao2015} does not explore the
effect of distribution on batching performance.

The G* system~\cite{Labouseur2015} manages graphs that correspond to
periodic snapshots, with the focus on efficient data layout.  It takes
advantage of the similarity between successive snapshots by storing
shared vertices only once and maintaining per-graph indexes.  Time is
not an intrinsic part of the system, as there is in \tga, and thus
temporal queries with time predicates like node creation are not
supported.  G* provides two query languages: procedural query language
PGQL, and a declarative graph query language (DGQL). PGQL provides
graph operators such as retrieving vertices and their edges from disk
and non-graph operators like aggregate, union, projection, and join.
All operators use a streaming model, i.e. like in traditional DBMS,
they pipeline.  DGQL is similar to SQL and is converted into PGQL by
the system.

Finally, the Historical Graph Store (HGS) system is an evolving graph
query system based on Spark.  It uses the property graph model and
supports retrieval tasks along time and entity dimensions through Java
and Python API.  It provides a range of operators such as selection
(equivalent to our subgraph operators but with no temporal
predicates), timeslice, nodecompute (similar to map but also with no
temporal information), as well as various evolution-centered
operators.  HGS does not provide formal semantics for any of the
operations it supports and the main focus is on efficient on-disk
representation for retrieval.

None of the three systems are publicly available, so direct
performance comparison with them is not feasible.

\section{Conclusions and Future Work}
\label{sec:conc}

In this paper we presented \tga: a tuple-stamped vertex and edge
relational model of evolving graphs and a rich set of operations with
point semantics.  TGA is TRA-vertex- and -edge-complete.  We show
reduction of each of our operations into TRA and further into RA with
additional primitives.

It is in our immediate plans to develop a declarative syntax for \tra,
making it accessible to a wider audience of users.  We described an
implementation of \ql in scope of Apache Spark, and studied performance
of operations on different physical representations.  Interestingly,
different physical implementations perform best for different
operations but support switching, opening up avenues for rule-based
and cost-based optimization.  Developing a query optimizer for \ql is
in our immediate plans.

\newpage

\bibliographystyle{abbrv}
\bibliography{temporal}

\begin{thebibliography}{10}

\bibitem{DBLP:conf/sigmod/AbergerTOR16}
C.~R. Aberger, S.~Tu, K.~Olukotun, and C.~R{\'{e}}.
\newblock Emptyheaded: {A} relational engine for graph processing.
\newblock In {\em Proceedings of the 2016 International Conference on
  Management of Data, {SIGMOD} Conference 2016, San Francisco, CA, USA, June 26
  - July 01, 2016}, pages 431--446, 2016.

\bibitem{DBLP:journals/csur/AggarwalS14}
C.~C. Aggarwal and K.~Subbian.
\newblock Evolutionary network analysis.
\newblock {\em {ACM} Comput. Surv.}, 47(1):10:1--10:36, 2014.

\bibitem{Amer-Yahia2009}
S.~Amer-Yahia, L.~Lakshmanan, and C.~Yu.
\newblock {SocialScope}: Enabling information discovery on social content
  sites.
\newblock In {\em CIDR}, 2009.

\bibitem{DBLP:reference/db/Bohlen09}
M.~H. B{\"{o}}hlen.
\newblock Temporal coalescing.
\newblock In {\em Encyclopedia of Database Systems}, pages 2932--2936. 2009.

\bibitem{DBLP:conf/vldb/BohlenSS96}
M.~H. B{\"{o}}hlen et~al.
\newblock Coalescing in temporal databases.
\newblock In {\em VLDB}, 1996.

\bibitem{DBLP:reference/db/Bohlen092}
M.~H. B{\"{o}}hlen et~al.
\newblock Temporal compatibility.
\newblock In {\em Encyclopedia of Database Systems}. 2009.

\bibitem{DBLP:conf/time/BohlenGJ06}
M.~H. B{\"{o}}hlen, J.~Gamper, and C.~S. Jensen.
\newblock How would you like to aggregate your temporal data?
\newblock In {\em TIME}, 2006.

\bibitem{Casteigts2011}
A.~Casteigts, P.~Flocchini, W.~Quattrociocchi, and N.~Santoro.
\newblock {Time-Varying Graphs and Dynamic Networks}.
\newblock In {\em Proceedings of the 10th international conference on Ad-hoc,
  mobile, and wireless networks}, volume 6811, pages 346--359, 2011.

\bibitem{Chan2008}
J.~Chan, J.~Bailey, and C.~Leckie.
\newblock {Discovering correlated spatio-temporal changes in evolving graphs}.
\newblock {\em Knowledge and Information Systems}, 16(1):53--96, 2008.

\bibitem{DBLP:conf/ictl/Chomicki94}
J.~Chomicki.
\newblock Temporal query languages: {A} survey.
\newblock In {\em Temporal Logic, First International Conference, {ICTL} '94,
  Bonn, Germany, July 11-14, 1994, Proceedings}, pages 506--534, 1994.

\bibitem{DBLP:reference/db/ChomickiT09b}
J.~Chomicki and D.~Toman.
\newblock Temporal relational calculus.
\newblock In {\em Encyclopedia of Database Systems}, pages 3015--3016. 2009.

\bibitem{Dignos2012}
A.~Dign{\"{o}}s, M.~H. B{\"{o}}hlen, and J.~Gamper.
\newblock {Temporal Alignment}.
\newblock In {\em Proceedings of the 2012 international conference on
  Management of Data - SIGMOD '12}, pages 433--444, Scottsdale, Arizona, USA,
  2012.

\bibitem{Ferreira2004}
A.~Ferreira.
\newblock {Building a reference combinatorial model for MANETs}.
\newblock {\em IEEE Network}, 18(5):24--29, 2004.

\bibitem{Gabielkov:2014:SSN:2591971.2591985}
M.~Gabielkov et~al.
\newblock Studying social networks at scale: Macroscopic anatomy of the twitter
  social graph.
\newblock In {\em SIGMETRICS}, 2014.

\bibitem{George2009}
B.~George, J.~M. Kang, and S.~Shekhar.
\newblock {Spatio-Temporal Sensor Graphs (STSG): A Data Model for the Discovery
  of Spatio-Temporal Patterns}.
\newblock {\em Intelligent Data Analysis}, 13(3):457--475, 2009.

\bibitem{George2006}
B.~George and S.~Shekhar.
\newblock {Time-aggregated graphs for modeling spatio-temporal networks}.
\newblock {\em Journal on Data Semantics}, 11:191--212, 2006.

\bibitem{DBLP:conf/osdi/GonzalezLGBG12}
J.~Gonzalez, Y.~Low, and H.~Gu.
\newblock Powergraph: Distributed graph-parallel computation on natural graphs.
\newblock In {\em OSDI}, pages 17--30, 2012.

\bibitem{DBLP:conf/osdi/GonzalezXDCFS14}
J.~E. Gonzalez et~al.
\newblock {GraphX}: Graph processing in a distributed dataflow framework.
\newblock In {\em OSDI}, 2014.

\bibitem{Hsu1995}
P.-y. Hsu and D.~S. Parker.
\newblock {Improving SQL with Generalized Quantifiers}.
\newblock In {\em ICDE}, 1995.

\bibitem{Huo2014}
W.~Huo and V.~J. Tsotras.
\newblock {Efficient Temporal Shortest Path Queries on Evolving Social Graphs}.
\newblock In {\em Proceedings of the 26th International Conference on
  Scientific and Statistical Database Management}, SSDBM '14, pages 38:1--38:4,
  New York, NY, USA, 2014. ACM.

\bibitem{DBLP:reference/db/JensenS09k}
C.~S. Jensen and R.~T. Snodgrass.
\newblock Temporal data models.
\newblock In {\em Encyclopedia of Database Systems}, pages 2952--2957. 2009.

\bibitem{Kan2009}
A.~Kan, J.~Chan, J.~Bailey, and C.~Leckie.
\newblock {A Query Based Approach for Mining Evolving Graphs}.
\newblock In {\em Eighth Australasian Data Mining Conference (AusDM 2009)},
  volume 101, Melbourne, Australia, 2009.

\bibitem{Khurana2013}
U.~Khurana and A.~Deshpande.
\newblock {Efficient Snapshot Retrieval over Historical Graph Data}.
\newblock In {\em 2013 IEEE 29th International Conference on Data Engineering
  (ICDE)}, pages 997 -- 1008, Brisbane, QLD, 2013.

\bibitem{Khurana2016}
U.~Khurana and A.~Deshpande.
\newblock {Storing and Analyzing Historical Graph Data at Scale}.
\newblock In {\em Proceedings of the 19th International Conference on Extending
  Database Technology, EDBT'16}, pages 65--76, Bordeaux, France, 2016.

\bibitem{DBLP:journals/sigmod/KulkarniM12}
K.~G. Kulkarni and J.~Michels.
\newblock Temporal features in {SQL:} 2011.
\newblock {\em {SIGMOD} Record}, 41(3):34--43, 2012.

\bibitem{Labouseur2015}
A.~G. Labouseur, J.~Birnbaum, P.~W. Olsen, S.~R. Spillane, J.~Vijayan, J.~H.
  Hwang, and W.~S. Han.
\newblock {The G* graph database: efficiently managing large distributed
  dynamic graphs}.
\newblock {\em Distributed and Parallel Databases}, 33(4):479--514, 2014.

\bibitem{Lerman2010}
K.~Lerman, M.~Rey, R.~Ghosh, and J.~H. Kang.
\newblock {Centrality Metric for Dynamic Networks}.
\newblock In {\em Proceedings of the Eighth Workshop on Mining and Learning
  with Graphs}, pages 70--77, Washington, D.C., 2010.

\bibitem{Li2005}
J.~Li et~al.
\newblock {Semantics and evaluation techniques for window aggregates in data
  streams}.
\newblock In {\em ACM SIGMOD}, 2005.

\bibitem{Miao2015}
Y.~Miao, W.~Han, K.~Li, M.~Wu, F.~Yang, L.~Zhou, V.~Prabhakaran, E.~Chen, and
  W.~Chen.
\newblock {ImmortalGraph: A System for Storage and Analysis of Temporal
  Graphs}.
\newblock {\em ACM Transactions on Storage}, 11(3):14--34, 2015.

\bibitem{Ren2011}
C.~Ren, E.~Lo, B.~Kao, X.~Zhu, and R.~Cheng.
\newblock {On Querying Historical Evolving Graph Sequences}.
\newblock {\em Proceedings of the VLDB Endowment}, 4(11):726--737, 2011.

\bibitem{GraphDB}
I.~Robinson, J.~Webber, and E.~Eifrem.
\newblock {\em Graph databases}.
\newblock O'Reilly Media, Inc., 2013.

\bibitem{Semertzidis2015}
K.~Semertzidis, K.~Lillis, and E.~Pitoura.
\newblock {TimeReach: Historical Reachability Queries on Evolving Graphs}.
\newblock In {\em Proceedings of the 18th International Conference on Extending
  Database Technology}, pages 121--132, Brussels, Belgium, 2015.

\bibitem{DBLP:conf/sigmod/SunFSKHX15}
W.~Sun, A.~Fokoue, K.~Srinivas, A.~Kementsietsidis, G.~Hu, and G.~T. Xie.
\newblock Sqlgraph: An efficient relational-based property graph store.
\newblock In {\em Proceedings of the 2015 {ACM} {SIGMOD} International
  Conference on Management of Data, Melbourne, Victoria, Australia, May 31 -
  June 4, 2015}, pages 1887--1901.

\bibitem{DBLP:reference/db/Toman09}
D.~Toman.
\newblock Point-stamped temporal models.
\newblock In {\em Encyclopedia of Database Systems}, pages 2119--2123. 2009.

\bibitem{Wood2012}
P.~T. Wood.
\newblock {Query languages for graph databases}.
\newblock {\em ACM SIGMOD Record}, 41(1):50--60, mar 2012.

\bibitem{DBLP:journals/pvldb/Xirogiannopoulos15}
K.~Xirogiannopoulos, U.~Khurana, and A.~Deshpande.
\newblock Graphgen: Exploring interesting graphs in relational data.
\newblock {\em {PVLDB}}, 8(12):2032--2035, 2015.

\bibitem{MoffittTempWeb16}
V.~{Zaychik Moffitt} and J.~{Stoyanovich}.
\newblock {Towards a distributed infrastructure for evolving graph analytics}.
\newblock In {\em TempWeb}, 2016.

\bibitem{DBLP:journals/sigmod/Zimanyi06}
E.~Zim{\'{a}}nyi.
\newblock Temporal aggregates and temporal universal quantification in standard
  {SQL}.
\newblock {\em {SIGMOD} Record}, 35(2):16--21, 2006.

\end{thebibliography}

\appendix 

\section{Additional examples.}
\label{sec:app:examples}

Figure~\ref{fig:tg_inter} shows the result of temporal intersection
of \insql{T1} with \insql{T2}.  Only the vertices and edges present in
both \tgs are produced, thus eliminating $v_3$ and $v_4$.  Period
$[2/15, 4/15)$ for $v_2$ is computed as a result of the join of
$[2/15, 5/15)$ in \insql{T1} and [$2/15, 4/15)$ in \insql{T2}.

Figure~\ref{fig:tg_diff} shows the result of temporal difference
of \insql{T1} with \insql{T2}.  Vertex v1 is removed between 2/15 and
6/15, splitting one v1 tuple in \tv of T1 into two temporally-disjoint
tuples in the result.

\begin{figure}[b]
\centering
\begin{subfigure}{3in}
\includegraphics[width=2.8in]{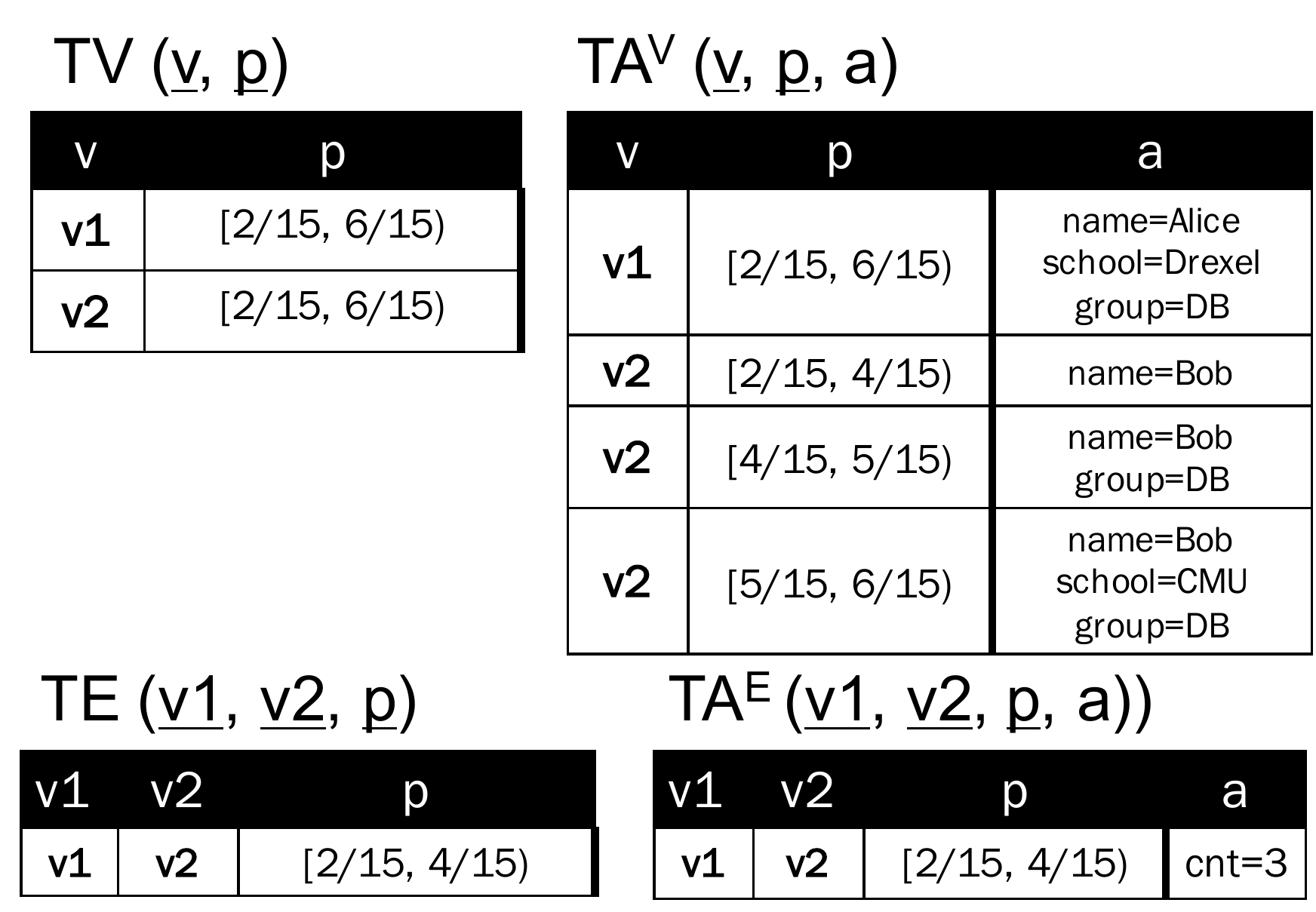}
\caption{$T1 \cap^T T2$.}
\label{fig:tg_inter}
\end{subfigure}
\begin{subfigure}{3in}
\includegraphics[width=2.8in]{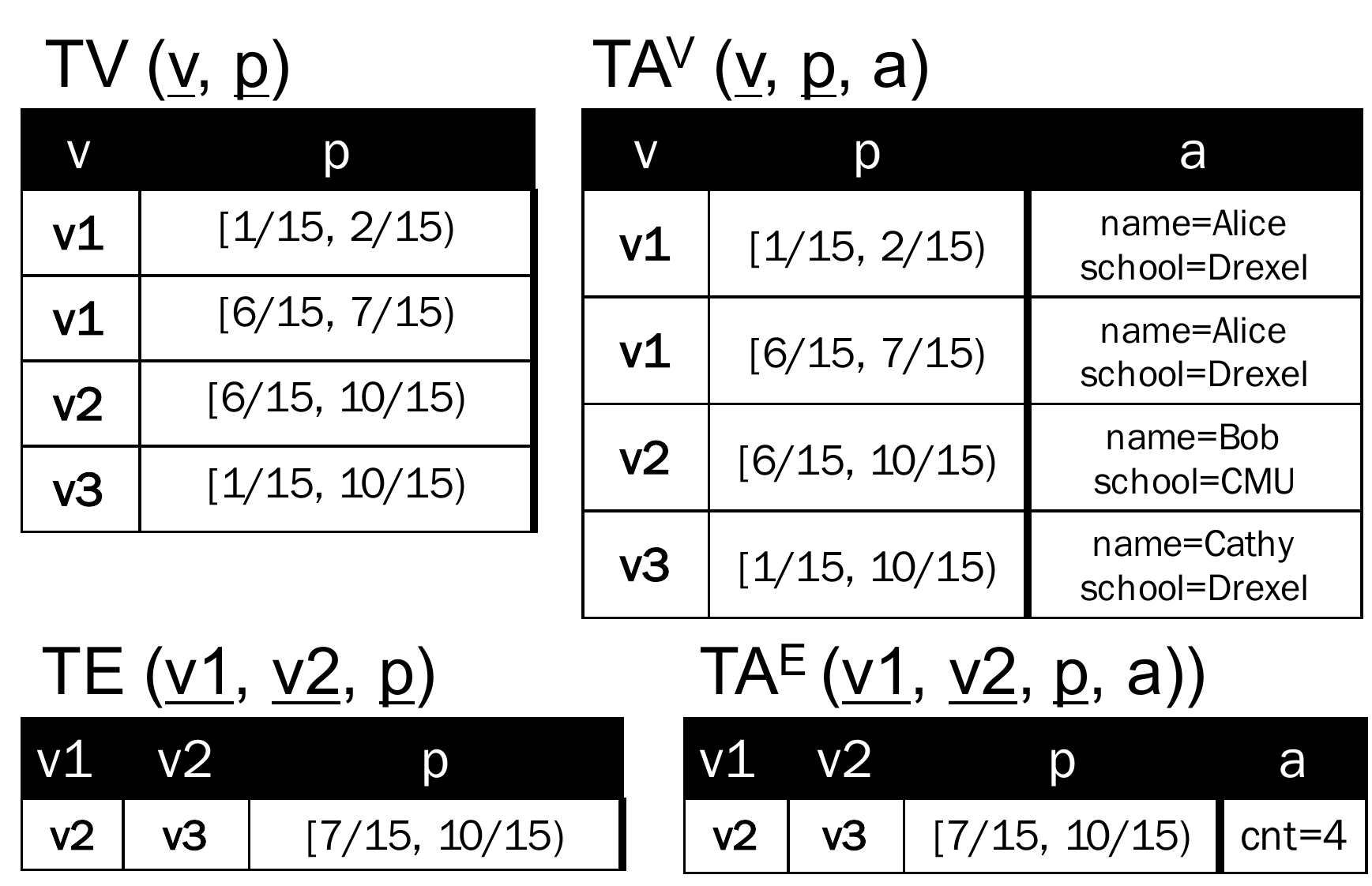}
\caption{$T1 \setminus^T T2$.}
\vspace{-0.2cm}
\label{fig:tg_diff}
\end{subfigure}
\caption{Binary operators.}
\label{fig:binary}
\vspace{-0.2cm}
\end{figure}

\section{Additional results.}
\label{sec:app2}

\begin{figure*}
\begin{minipage}[b]{2.2in}
\centering
\includegraphics[width=2.1in]{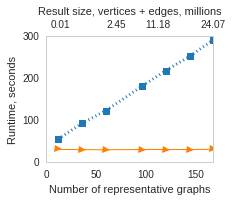}
\caption{Slice on wiki-talk.}
\label{fig:slicewiki}
\end{minipage}
\begin{minipage}[b]{2.2in}
\centering
\includegraphics[width=2.1in]{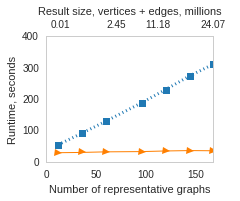}
\caption{Map on wiki-talk.}
\label{fig:project}
\end{minipage}
\begin{minipage}[b]{2.2in}
\centering
\includegraphics[width=2.4in]{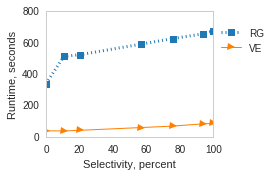}
\caption{Subgraph on wiki-talk.}
\label{fig:subgraphwiki}
\end{minipage}
\end{figure*}

\begin{figure*}
\centering
\begin{minipage}[b]{2in}
\centering
\includegraphics[width=2in]{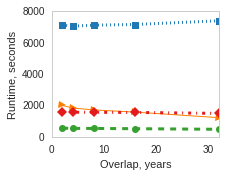}
\caption{Union on nGrams.}
\label{fig:union2}
\end{minipage}
\begin{minipage}[b]{2.3in}
\centering
\includegraphics[width=2.3in]{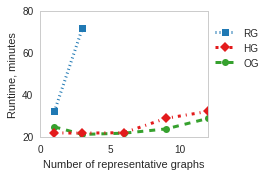}
\caption{PageRank on Twitter.}
\label{fig:pranktwitter}
\end{minipage}
\begin{minipage}[b]{2.3in}
\centering
\includegraphics[width=2.3in]{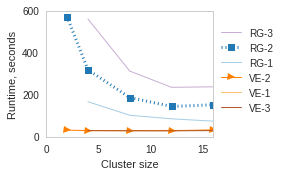}
\caption{Scaling slice on wiki-talk.}
\label{fig:slicescale}
\end{minipage}
\end{figure*}

\begin{figure*}
\centering
\begin{minipage}{2.4in}
\includegraphics[width=2.4in]{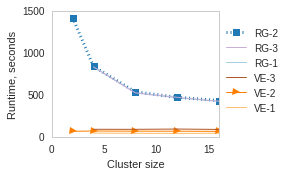}
\caption{V-subgraph on wiki-talk.}
\label{fig:selectscale}
\end{minipage}
\end{figure*}

Plots and discussion in this section complement experimental results
presented in Section~\ref{sec:exp}.

Figure~\ref{fig:slicewiki} shows performance of VE and \sg
on \insql{slice} over wiki-talk.  It exhibits the same trend as on the
nGrams dataset in Figure~\ref{fig:slicengrams} but on a smaller scale.

Figure~\ref{fig:subgraphwiki} shows performance of VE and \sg
on \insql{vertex subgraph} over wiki-talk.  Wiki-talk dataset is small
enough that broadcast join can be used for constraining the edges, so
the sudden worsening of performance is not observed, as it is in
Figure~\ref{fig:subgraphngrams}.

We next examine how the different access methods scale with the size
of the cluster.  We varied the number of cluster workers while
executing individual operations.

To examine the performance on slice, we fixed the slice interval size
to be 4, 8, and 14 years (series -1, -2, and -3, respectively).  As
can be seen in Figure~\ref{fig:slicescale}, the performance of VE did
not change with the cluster size or the size of the interval.  Our
cluster stores the graph files on HDFS with no replication, and Spark
does not currently support filter pushdown on dates (this is being
addressed in one of the upcoming releases), so these results are
expected as the operation is essentially just a file scan.  \rg
performance did improve as the cluster grew, with the biggest
reduction occurring by 8 slaves and diminishing returns thereafter. 

Similar trends can be seen on vertex subgraph in
Figure~\ref{fig:selectscale}, where we fixed the query selectivity to
be 20, 57, and 100\%.  There is no observable difference between
different selectivity for \rg, which is consistent with the subgraph
experiment results.

We do not include results for every operation here as they all show
the same trends -- performance rapidly improves with increased cluster
size up to a point and adding additional slaves is not beneficial
thereafter.  \eat{The final result in Figure~\ref{fig:ccscale} shows
  the performance on connected components analytic with 2, 5, and 8
  years.  The very small cluster sizes led to executions that were
  stopped due to taking too long.}

\end{document}